\documentclass[12pt]{article}
\usepackage{geometry,float}
\usepackage{amssymb}
\usepackage{amstext}
\usepackage{bm}
\usepackage{graphics}
\usepackage{graphicx}
\usepackage{amsthm}
\usepackage{amsmath}
\usepackage{latexsym}
\usepackage{rotate}
\usepackage{lscape}
\usepackage{color}
\usepackage{epsfig,epsf,psfrag}
\usepackage{sectsty}
\usepackage{graphics}
\usepackage{epsfig}
\usepackage{multirow}
\usepackage{graphicx,hyperref}
\usepackage{amsmath}
\usepackage{amsthm}
\usepackage{amssymb}
\usepackage{sectsty}
\usepackage{epsfig,natbib}
\usepackage{rotate}
\usepackage{lscape}
\usepackage{setspace}
\usepackage{subcaption}
\usepackage{float}
\usepackage{graphicx}
\usepackage{mathtools}
\usepackage{enumitem}
\usepackage{xr}
\newtheorem{Theorem}{Theorem}
\newtheorem{Lemma}{Lemma}

\def\references{\bibliography{1-1-19-1}}

\def\red{\textcolor{red}}

\sectionfont{\centering\bf\sf\large}
\subsectionfont{\sf\normalsize}

\renewcommand{\baselinestretch}{1.6} 
\newcommand{\single}{\renewcommand{\baselinestretch}{1.2}\normalsize}
\newcommand{\double}{\renewcommand{\baselinestretch}{1.63}\normalsize}

  \renewenvironment{thebibliography}[1]{
    \begin{oldthebibliography}{#1}
      \setlength{\parskip}{0ex}
      \setlength{\itemsep}{0ex}
  }
  {
    \end{oldthebibliography}
  }

\newcommand{\bea}{\begin{eqnarray*}}
\newcommand{\eea}{\end{eqnarray*}}
\newcommand{\be}{\begin{eqnarray}}
\newcommand{\ee}{\end{eqnarray}}
\newcommand{\ed}{\end{document}}

\newcommand{\btab}{\begin{tabular}}
\newcommand{\etab}{\end{tabular}}

\newcommand{\la}{\label}
\newcommand{\bi}{\begin{itemize}}
\newcommand{\ei}{\end{itemize}}
\newcommand{\bfi}{\begin{figure}}
\newcommand{\efi}{\end{figure}}
\newcommand{\ben}{\begin{enumerate}}
\newcommand{\een}{\end{enumerate}}
\newcommand{\bay}{\begin{array}}
\newcommand{\eay}{\end{array}}

\def\sv{V^\ast(t)}

\def\Vi{V^\ast_i}
\def\Wi{V_i}
\def\F{Fr\'{e}chet}
\def\o{\omega}
\def\O{\Omega}
\def\bco{\iffalse}

\def\var{{\rm var}}

\def\ci{\cite}
\def\cp{\citep}
\def\eps{\varepsilon}

\def\om{\, \o \in \O:\,}

\newcommand{\no}{\noindent}
\newcommand{\bc}{\begin{center}}
\newcommand{\ec}{\end{center}}

\DeclareMathOperator*{\argmin}{argmin}

\bibpunct{(}{)}{;}{a}{}{,}

\begin{document}
\thispagestyle{empty} \single \bc {\bf \sc \Large Modeling  Time-Varying Random Objects and Dynamic Networks\footnote{Research supported by NSF Grants DMS-1712864 and DMS-2014626.}}

\vspace{1cm}

Paromita Dubey$^{\dagger}$ and  Hans-Georg M\"uller$^{\dagger\dagger}$  \\
$^{\dagger}$Department of Statistics, Stanford University\\
$^{\dagger\dagger}$Department of Statistics, University of California, Davis \ec \centerline{April 2021}

\vspace{0.1in} \thispagestyle{empty}
\bc{\bf \sf ABSTRACT} \ec \vspace{-.1in} \no 
 Samples of dynamic or time-varying networks and other random object data such as time-varying probability distributions are increasingly encountered in modern data analysis. Common methods for 
 time-varying data such as functional data analysis are infeasible when observations are time courses of networks or other complex non-Euclidean random objects that are elements of  general metric spaces.  In such spaces, only pairwise distances between the data objects are available and a strong limitation is that  one cannot carry out arithmetic operations due to the lack of an algebraic structure. We combat this complexity by a generalized notion of mean trajectory taking values in the object space. For this, we  adopt pointwise \F \ means and then construct  pointwise distance trajectories between the individual time courses and the estimated \F \  mean trajectory, thus representing the time-varying objects and networks by functional data.   Functional principal component analysis of these  distance trajectories can reveal interesting features of dynamic networks and  object time courses and is useful for downstream analysis. Our approach also makes it possible to study the  empirical dynamics of time-varying objects,  including 
   dynamic regression to the mean or explosive behavior over time.  
 We demonstrate desirable asymptotic properties of  sample based estimators for suitable population targets under mild assumptions.  The utility of the proposed methodology is illustrated with {dynamic networks,  time-varying distribution data and longitudinal growth data.}

\vspace{0.1  in}
\no {KEY WORDS: Empirical dynamics; \F \ mean trajectory; Functional data analysis;  Metric space; Object time courses; Time-varying distributions; Time-varying networks.  
\thispagestyle{empty} \vfill

\newpage
\pagenumbering{arabic} \setcounter{page}{1} \double
\section{INTRODUCTION}
Longitudinal or time-varying data consist of repeated observations for each subject at different time points,  where one has such observations for a sample of independent subjects. Tools for analyzing both univariate and multivariate longitudinal data have been well studied \cp{fitz:08,verb:14,fieu:06,berr:11,zhou:08,xian:13}.  When such data are scalar or Euclidean vectors and densely measured in time,  they can be analyzed as  functional data \cp{rice:04,guo:04:1,yang:07}, where functional data analysis provides a flexible nonparametric framework with  a well established  toolbox. This popular methodology requires a vector space structure of the observed data and is therefore restricted to the case where the  measurements at each fixed time are scalars or Euclidean  vectors \citep{rams:05,ferr:07,horv:12,hsin:15,mull:16:3}.

Time-varying object data are becoming more frequent    and are encountered in  time-varying social networks, traffic networks that change over time,  brain networks between hubs that evolve with age, and many other settings \cp{nie:17}. While such data have  similarities with  densely measured functional data,  the observations at each time point are neither scalars nor vectors as in classical functional data analysis, but instead take values in a general metric space. A major challenge is that in such spaces typical vector space operations such as  addition, scalar multiplication or inner products are not defined.  In general metric spaces, the only information available are  pairwise distances between the random objects at each observation time, and therefore the tools of functional and longitudinal data analysis are not directly applicable.  We develop here a simple method that bypasses  these challenges by focusing on distances as outcomes. 

Models for time courses of non-Euclidean objects have been developed for shape evolution as a continuous diffeomorphic deformation of a baseline shape over time 
\citep{durr:13} and  for the analysis of longitudinal data taking values in smooth Riemannian manifolds \citep{schi:15, mura:15,anir:17,mull:18:4}. These methods  exploit the local Euclidean nature of Riemannian manifolds but are not applicable for the analysis of data objects in more general metric spaces that do not have a natural Riemannian geometry. 
For this general case, a fairly complex and challenging  methodology has been developed in \ci{mull:19:1}. 
We propose here a more straightforward approach that allows us to cut through the challenges 
posed by longitudinal metric space valued data by focusing on the distance of time-varying random objects from the mean trajectory and thereby reducing such data to classical functional data. 

A related topic is 
the analysis of longitudinal functional data, where the observations at each time point are functional rather than scalar. For this scenario, previous approaches  \citep{mull:17:4} have utilized a  tensor product representation of the function-valued stochastic process, an approach for which an underlying  Hilbert space is essential and that cannot be directly extended to non-Hilbertian data that we consider here. While complex longitudinal network data have been extensively studied \citep{snij:05,huis:03,koss:06},   these efforts have been  directed specifically towards studying dynamics of evolution of a single network with a variety of network effects and are not applicable to the study of the dynamics of a sample of network trajectories or longitudinal trajectories of other general data objects.

Our goal in this paper is to provide a straightforward methodology  for analyzing {\it functional object data}, i.e., time-varying random objects including  dynamic networks that live  in general metric spaces. Converting such data to functional data makes it possible to tap the rich existing toolbox of functional data analysis, while   imposing not more than  mild entropy conditions on the underlying object space. We assume that the data objects take values in a totally bounded metric space and that  the random object trajectories are fully observed. Since many dynamic developments of interest  can be expressed in terms of departure of an observed  dynamic process from a baseline process, we use a suitably defined  mean trajectory that takes values in the object space as baseline. 

For metric space valued data, \F \ developed a generalization of the usual population and sample means \citep{frec:48}, which then  gives rise to a generalization of the notion of variance for object data, quantifying the variation of such data  around the \F \ mean. It is thus natural to take as  mean trajectory of the object functional data the trajectory consisting of the pointwise \F \ means,  obtained at each time point. In  order to study individual deviations of each time course from this mean trajectory, we first construct squared distance trajectories of the individual object functions from the \F \ mean trajectory. These squared distance trajectories, which we refer to as subject specific \F \ variance trajectories, are scalar valued functional data  in contrast  to the object trajectories themselves,  and therefore  can be subjected to the standard tools that have been developed for functional data analysis, including the highly successful functional principal component analysis \cp{klef:73,hall:06:1,li:14,chen:15:1,lin:16}. 

One major obstacle in working with subject specific \F \ variance trajectories, which makes it difficult to directly apply functional data methods, is that the population \F \ mean trajectory is not known and has to be estimated from the data. The squared distance trajectories which are used for the analysis are then the squared distances of the individual time courses from the sample \F \ mean trajectory. This makes the subject-specific \F \ variance trajectories dependent and so they cannot be treated as independent observations of random trajectories, which is essential for the application of the usual  functional data analysis tools.  By 
imposing mild assumptions on the entropy of the underlying metric space and on the continuity of the random object trajectories, 
we are able to overcome this problem and to 
establish desirable asymptotic properties of the estimators of suitably defined population targets, including  rates of convergence.

As we demonstrate in various applications, functional principal component analysis of the subject-specific \F \ variance trajectories can lead to interesting insights regarding the behavior of the object trajectories. Clustering is often a useful first step  for  exploratory data analysis, aiming to  identify homogeneous subgroups and patterns that have some meaningful interpretation for the researcher. Functional data are inherently infinite-dimensional and a probability density generally does not exist, which contributes to the challenge of  clustering functional data  \citep{chio:07,jacq:14, tarp:03,ciol:16,suar:16}. An additional difficulty arises when  the observations at each time point are not in a vector space. Eigenfunctions of the covariance surface of the \F \ variance functions can nevertheless pinpoint  predominant modes of variation of the individual \F \ variance trajectories around the average \F \ variance trajectory and projection scores of the subject-specific \F \ variance trajectories along these eigenfunctions can reveal inherent clustering within the object trajectories, as we will illustrate in our data applications. 
 
 
 Identifying extremes and potential outliers is challenging for functional data because the observations  at a given time point itself may not be unusual in their value  but the overall shape of an observed curve may be very different from that of the bulk of curves.  Object time courses are even more intractable. Statistical data depth is a concept introduced to measure the ``centrality"  or the ``outlyingness"  of an observation within a given data set or an underlying distribution and this concept has been extended to functional data in recent years \citep{lope:09, nagy:17, nagy:19, agos:18}, where it has been used widely for the detection of extremes and potential  outliers in functional data \citep{ren:17,febr:08,arri:14,roma:13}. 
 
 An important aspect of our analysis is that conclusions about the behavior of object time courses are drawn based on their squared distances from the \F \ mean trajectory. The \F \ mean trajectory is a representative for the most central point for a sample of object functions and the subject specific \F \ variance time courses carry information about the deviations of individual trajectories from the `central' trajectory,  which leads to a central-outward ordering for the sample trajectories. We show that principal component projection scores of the subject specific \F\ variance time courses along eigenfunctions are useful for visualization of the longitudinal object data and also for the detection of extremes. 
Another aspect of interest is the dynamics of the evolving object trajectories, especially whether they tend to move closer to the mean function as time progresses, so that a far away trajectory will tend to be drawn towards the center (centripetality)  or will move further away from the center (centrifugality),  as time progresses.

The paper is organized as follows: In Section 2,  we introduce our  framework and  define the  population targets and the corresponding sample based estimators. The theoretical properties of the estimators are established in  Section  3. {This is  followed by  data illustrations in Section 4,   where we apply the  proposed method for the longitudinal network generated by the Chicago Divvy bike data for the years 2014 to 2017, for  the longitudinal  annual fertility data for  26 countries over  34 calendar years from 1976 to 2009 and for time varying shape data using the Z\"urich longitudinal growth study. We  also demonstrate  the proposed quantification of the underlying dynamics of the observed processes.}  Simulation results for a  sample of time-varying networks are presented in  Section 5, followed by a discussion in Section 6. Auxiliary results and proofs can be found  in the online supplement.

\section{PRELIMINARIES AND ESTIMATION}
\label{sec: prelim}
We consider an object space $(\O,d)$ that is a  totally separable bounded metric space and an $\O$-valued stochastic process  $\{X(t)\}_{t \in [0,1]}$, alternatively referred to as $X$ for ease of notation, and assume that   
one observes  a sample of random object trajectories $X_1, X_2,\dots,X_n$, which are independently and identically distributed copies of the random process $X$ with respect to an underlying probability measure $P$.
For each subject-specific trajectory $X_i$, we aim to quantify  its deviation from a baseline object function, which can be thought of as the mean or  typical population trajectory. A natural baseline for real-valued functional data is the mean function. For more general general object-valued trajectories,  we propose to use the population \F \ mean trajectory as baseline function, defined as the pointwise \F \ mean function,  where for given $t \in [0,1]$, the population  and sample \F \ mean trajectories at $t$ are defined as
\be  \la{mean}
\mu(t)=\argmin_{\o \in \O} E\left(d^2(X(t),\o ) \right), \quad 
\hat{\mu}(t)=\argmin_{\o \in \O} \frac{1}{n} \sum_{i=1}^{n}d^2(X_i(t),\o),
\ee respectively.  Here we assume that for all $t \in [0,1]$ these minimizers exist and are unique. {While the existence and uniqueness of \F \ means is not guaranteed in general spaces \citep{bhat:03}, for the case of Hadamard spaces, which have globally nonpositive curvature,  \F \ means as defined in equation \eqref{mean} exist and are unique \cp{stur:03}. For positively curved spaces see \cite{ahid:20}.}

The target functions for our analysis then ideally would be the  functions
\be \la{or} V^\ast_i(t)=d^2(X_i(t),\mu(t)),  \,\, t \in [0,1],\ee 
which correspond to  the pointwise squared distance functions of the subject trajectories $X_i$ from the population \F \ mean function $\mu=\mu(t)$ for the subject trajectories $X_i$.  These can be characterized as the subject-specific oracle \F \ variance trajectories. They are however unavailable, since   the population \F \ mean trajectory $\mu$ is unknown and needs  to be estimated from the data. From   \eqref{or} one obtains the  data-based version 
\be \la{db} V_i(t)=d^2(X_i(t),\hat{\mu}(t)),  \, t \in [0,1],\ee 
where $\hat{\mu}$ is as in \eqref{mean} and we refer to the $V_i=V_i(t)$ as the  sample \F \ variance trajectories and write  $V(t)=d^2(X(t),\hat{\mu}(t))$ for the generic version. Since they all depend on $\hat{\mu}(t)$, the sample \F \ variance trajectories are dependent and cannot be treated as independent realizations of a stochastic process, which is the standard framework for functional data analysis, thus posing a challenge for theory. 

Suppose for the moment that we have available an i.i.d. sample of oracle  \F \ variance trajectories $V^\ast_i$, with generic version denoted by $V^\ast$. Then a typical dimension reduction step in Functional Data Analysis (FDA) is to apply  Functional Principal Component Analysis (FPCA), which 
facilitates the conversion of the  functional data $V^\ast_i$ to a countable sequence of uncorrelated random variables, the functional principal components (FPCs),  where the sequence of FPCs is often  truncated at a finite dimensional random vector to achieve  dimension reduction. 
FPCA is based on using the eigenfunctions of the auto-covariance operator of the process $V^\ast$. This  is an integral operator,  a trace class and moreover compact Hilbert Schmidt operator \citep{hsin:15} that  has the population \F \ covariance surface $C$ as its kernel, where 
\begin{equation} \la{C}
C(s,t)=E\left(d^2(X(s),\mu(s)) d^2(X(t),\mu(t)) \right)-E\left(d^2(X(s),\mu(s)) \right) E\left(d^2(X(t),\mu(t)) \right).
\end{equation}

The eigenvalues of the auto-covariance operator are nonnegative as the covariance surface is symmetric and nonnegative definite. By Mercer's theorem,
\begin{equation*}
C(s,t)=\sum_{j=1}^{\infty} \lambda_j \phi_j(s) \phi_j(t), \quad   s, t \in [0,1],
\end{equation*}
with uniform convergence, where the $\lambda_j$ are the eigenvalues of the covariance operator, ordered in decreasing order, and $\phi_j(\cdot)$ are the corresponding orthonormal eigenfunctions. 

This leads to the Karhunen-Lo\`{e}ve expansion of the oracle \F \ variance trajectories,
\begin{equation} \la{klo}
V^\ast_i(t) = \nu^\ast(t)+\sum_{j=1}^{\infty} B_{ij} \phi_{j}(t),
\end{equation}
with $L^2$ convergence. Here  $\nu^\ast$ is the mean function of the subjectwise  \F \ variance functions, $\nu^\ast(t)= E(d^2(X(t),\mu(t))),$ the 
population \F \ variance function.  The  $B_{ij}$ are the FPCs,  which are uncorrelated across $j$ with $E(B_{ij})=0$, ${\rm var}(B_{ij})=\lambda_j$ and
$B_{ij}=\int (V^\ast_i(t)-\nu^\ast(t)) \phi_j(t)dt.$
If $\mu$ is known, according to \eqref{C}, the oracle estimator of the \F \ covariance surface is 
\begin{equation} \la{Ct}
\tilde{C}(s,t)=\frac{1}{n} \sum_{i=1}^{n} \Vi(s)\Vi(t) -\frac{1}{n} \sum_{i=1}^{n} \Vi(s)\,\,\frac{1}{n}  \sum_{i=1}^{n} \Vi(t).
\end{equation}
Under mild assumptions  on the functional trajectories $V_i^*$,   standard asymptotic theory from functional data analysis shows that  this estimator  has desirable asymptotic properties and converges to the true covariance surface $C$  \eqref{C} \cp{hall:06:1}.

As  the population \F \ mean function $\mu$ in reality is however unknown, we need to replace it by the  sample based estimator $\hat{\mu}$  of the \F \ covariance surface, 
\begin{equation} \la{Ch}
\hat{C}(s,t)=\frac{1}{n} \sum_{i=1}^{n} \Wi(s)\Wi(t) -\frac{1}{n} \sum_{i=1}^{n} \Wi(s)\,\,\frac{1}{n}  \sum_{i=1}^{n} \Wi(t),
\end{equation}
using the data-based distance processes $\Wi(t)$ (\ref{db}) that depend on  estimates  $\hat{\mu}$. 
We show in Section 3  that $\hat{C}$ is asymptotically close to the oracle estimator of the \F \ covariance surface $\tilde{C}$  under mild regularity conditions on the metric space and the object functions and therefore has desirable asymptotic properties as an estimator of the population \F \ covariance surface. 

Estimates of  eigenvalues and eigenfunctions are obtained as the empirical eigenvalues and eigenfunctions of the integral covariance operator with covariance 
kernel  $\hat{C}$, and will be denoted by $\hat{\lambda}_j$ and $\hat{\phi_j}$, ordered  in decreasing order of the  eigenvalues.  Eigenfunctions $\phi_j$ can be interpreted as coordinate directions, thereby providing the basis for principal modes of variation of the subject specific oracle \F \ variance trajectories around the population \F \ variance trajectory. Modes of variation \cp{cast:86,lila:19} are useful to  quantify the departure of a random object trajectory from the \F \ mean function.  The eigenfunctions can be viewed as {\it modes  of  outlyingness}  of the subject-specific trajectories. The estimates of the projection scores $B_{ij}$ of the $i^{th}$ oracle \F \ variance trajectory on  the $j^{th}$ 
eigenfunction are given by
\begin{equation} \la{pre} 
\hat{B}_{ij}=\int_{0}^{1} \left(V_i(t)-\frac{1}{n}\sum_{k=1}^{n} V_k(t)\right) \hat{\phi}_j(t) dt.
\end{equation}
We show in  section \ref{sec: theory}  that under regularity assumptions, the   
$\hat{B}_{ij}$ are asymptotically close to the $B_{ij}$ in (\ref{klo}). These scores are useful for visualizing common traits in object trajectories and for detecting  extremes,  homogeneous subgroups or clusters in the data.

\section{THEORY}
\label{sec: theory}

We establish  asymptotic properties of the empirical estimators of the population targets as described in section \ref{sec: prelim}, assuming that  the realizations of the object valued random process $X(t),\,\, t \in [0,1]$,  have continuous sample paths almost surely and  take values in a totally bounded metric space $(\O,d)$ and the following conditions are satisfied: 
\begin{itemize}
	\item[(A1)] The objects $\mu(s)$ and $\hat{\mu}(s)$ exist and are unique, the latter almost surely, for each $s \in [0,1]$. Additionally for any $\epsilon > 0$,
	\begin{equation*}
	\inf_{s \in [0,1]} \inf_{\om d(\o,\mu(s)) > \epsilon} E(d^2(X(s),\o))-E(d^2(X(s),\mu(s))) > 0
	\end{equation*}
	and there exists a $\tau=\tau(\epsilon) > 0$ such that
	\begin{equation*}
	\lim_{n \rightarrow \infty}P \left(\inf_{s \in [0,1]} \inf_{\om d(\o,\hat{\mu}(s)) > \epsilon} \frac{1}{n} \sum_{l=1}^{n} \lbrace d^2(X_i(s),\o)-d^2(X_i(s),\hat{\mu}(s))\rbrace \geq \tau(\epsilon) \right) = 1.
	\end{equation*}
	
	
		\item[(A2)]There exists $\rho> 0, \ D > 0$ and $\beta > 1$ such that 
	\begin{equation*}
	\inf_{s \in [0,1]} \inf_{\om d(\o,\mu(s)) < \rho} \lbrace E(d^2(X(s),\o))-E(d^2(X(s),\mu(s))) -D d^{\beta}(\o,\mu(s))\rbrace \geq 0.
	\end{equation*}
	
	\item[(A3)] For some $0 < \alpha \leq 1$,  the random function $X(\cdot)$ defined on $[0,1]$ and taking values in $\O$, where we denote the space of all such functions as $\O^{[0,1]}$,  is $\alpha$-H\"{o}lder continuous,  
	i.e., for nonnegative $G: \O^{[0,1]}\rightarrow \mathbb{R}^{+}$ with  $E\left(G(X)^2\right) < \infty$ , it holds almost surely,
	\begin{equation*}
	d(X(s),X(t)) \leq G(X) |s-t|^\alpha.
	\end{equation*}
	
	\item[(A4)]	
	For $I(\delta)=\int_{0}^{1} \sup_{s \in [0,1]} \sqrt{\log N({\eps\delta},B_{\delta}(\mu(s)),d)} d\eps $ it holds that  $I(\delta)=O(1)$ as $\delta \rightarrow 0$.  Here $B_{\delta}(\mu(s))=\{\o \in \O: d(\o,\mu(s)) < \delta\}$ is the $\delta$-ball around $\mu(s)$ and $N(\gamma,B_{\delta}(\mu(s)),d)$ is the  covering  number, i.e.,  the minimum number of balls of radius $\gamma$ required to cover $B_{\delta}(\mu(s))$ 	\cp{well:96}. 

\end{itemize} 

 Assumption (A1) guarantees uniform convergence of the sample \F \ mean trajectory to its population target as it implies $\sup_{s \in [0,1]} d(\hat{\mu}(s),\mu(s))=o_P(1)$ \cp{mull:19:1}; for convenience, we  state this result as  Lemma  \ref{lma: mean_uniform} in the Supplement without proof. Measurability issues of the sample \F \ mean function can be dealt with in a similar fashion as $M$-estimators in general by considering outer probability measures; for more detailed  discussion of the  measurability issues see sections 1.2, 1.3 and 1.7 of \cite{well:96}. Assumptions of type (A2) are standard for M-estimators and characterize the local curvature of the target function to be minimized near the  minimum; this curvature is characterized by $\beta$, which features in the resulting rate of convergence. Lemma \ref{lma: rate} in section \hyperref[supp1]{A.3} of the Supplement provides the rate of convergence of the \F \ mean function $\hat{\mu}(s)$, and corrects an algebraic error in Theorem 3 in  \cite{mull:19:1}. {When $(\O,d)$ is a Hadamard space, $\beta$ takes the value 2 \citep{stur:03} for any probability measure on $\O$, and therefore assumptions (A1) and (A2) are satisfied. } 
 	
Assumptions  (A3) and (A4) are required for measuring the size of the space of object functions and imply  an entropy condition on the object function space, which then leads to uniform convergence of the plug-in estimator of the \F \ covariance surface $\tilde{C}$ in \eqref{Ct} given by $\hat{C}$ \eqref{Ch} at a fast rate.  In (A3), we assume that the rate of H\"{o}lder continuity of the random object trajectories is fixed,  with the H\"{o}lder constant having a finite second moment, which means that $E  \left \lbrace \left( \sup_{s \neq t, s,t \in [0,1]}d(X(s),X(t))/|s-t|^{\alpha} \right)^2 \right \rbrace $ is finite. This assumption is a mild smoothness assumption satisfied for certain values of $\alpha$  by many common Euclidean-valued random processes, including the Wiener process for $\alpha =\frac{1}{2}$. 

Assumption (A3) together with the curvature condition in (A2) implies H\"older continuity of the \F \ mean function $\mu(\cdot)$. For details, we refer to the proof of Lemma \ref{lma:entropy} in section \hyperref[supp1]{A.3} of the Supplement. 
Assumption (A4) is a bound on the covering number of the object metric space and is satisfied by several commonly encountered  random objects, including random probability distributions equipped with the 2-Wasserstein metric, covariance matrices of fixed dimension and graph Laplacians of networks with fixed number of nodes \cp{mull:18:5,mull:19:3}. We provide a proof and further discussion on this in section \hyperref[supp2]{A.5} of the Supplement, where we show that   the space of univariate distributions with the 2-Wasserstein metric and the space of graph Laplacians with the Frobenius metric  satisfy assumptions (A1)-(A4). Theorem \ref{lma: gauss} below gives the uniform convergence of  $(s,t) \mapsto \sqrt{n} (\hat{C}(s,t)-C(s,t))$.

\begin{Theorem} 
	\label{lma: gauss}
	Under assumptions (A1)-(A4), 
    \begin{equation*}
		 \sup_{s,t \in [0,1]}\left|\hat{C}(s,t)-{C}(s,t)\right|=O_P\left(\max \left\lbrace\left(\frac{\sqrt{\log{n}}}{n}\right)^{\frac{1}{2(\beta-1)}},\frac{1}{\sqrt{n}}\right\rbrace \right).
\end{equation*}
\end{Theorem}

This convergence result provides the major justification that observed processes $V$ may be used for the proposed FPCA instead of the oracle processes $V^\ast$ when the mean function has to be estimated, as will invariably be the case in practical applications. 
Uniform convergence and rates of convergence of $|\lambda_j-\hat{\lambda}_j|$ and $\sup_{s \in [0,1]} \left| \hat{\phi}_j(s)-\phi_j(s)\right|$ then  follow  from Theorem \ref{lma: gauss} by standard perturbation results along the lines of  the Davis-Kahan theorem,  e.g., Lemma 4.3 of  \cite{bosq:00}. Our next  result provides a quantification of the  asymptotic closeness of the sample based estimators of the FPCs and the oracle FPCs and requires the additional assumption 

\begin{itemize} 

	\item[(A6)]For each $j \geq 1$, the eigenvalue $\lambda_j$, as defined in section \ref{sec: prelim}, has multiplicity 1, i.e., it holds that $\delta_j > 0$ where $\delta_j=\min_{1\leq l\leq j}(\lambda_l-\lambda_{l+1})$.
\end{itemize}

\begin{Theorem}
	\label{lma: eigen}
	Under assumptions (A1)-(A6), we have 
\bea
|\hat{\lambda}_j-\lambda_j|&=&O_P\left(\max \left\lbrace\left(\frac{\sqrt{\log{n}}}{n}\right)^{\frac{1}{2(\beta-1)}},\frac{1}{\sqrt{n}}\right\rbrace \right),\\
\sup_{s \in [0,1]} \left| \hat{\phi}_j(s)-\phi_j(s)\right| &=&O_P\left(\frac{1}{\delta_j}\max \left\lbrace\left(\frac{\sqrt{\log{n}}}{n}\right)^{\frac{1}{2(\beta-1)}},\frac{1}{\sqrt{n}}\right\rbrace \right),\\
|\hat{B}_{ij}-B_{ij}|&=&O_P\left(\max\left\lbrace \frac{1}{\delta_j},1\right\rbrace \max \left\lbrace\left(\frac{\sqrt{\log{n}}}{n}\right)^{\frac{1}{2(\beta-1)}},\frac{1}{\sqrt{n}}\right\rbrace \right).
\eea
\end{Theorem}
Examples of object spaces that satisfy the assumptions  include graph Laplacians of connected, undirected and simple graphs corresponding to networks of fixed dimension equipped with the Frobenius metric \citep{mull:18:5}, as well as  univariate probability distributions equipped with the $2$-Wasserstein metric and  correlation matrices of a fixed dimension equipped with the Frobenius metric  \citep{mull:19:3}. For all of the above examples, one has $\beta=2$ in assumption (A2). For a detailed discussion, see Section \hyperref[spaces]{A.5} in the Supplement.

\section{DATA ILLUSTRATIONS}
\label{sec: data}

\bco
\subsection{New York Taxi Cab Data}
The New York City Taxi and Limousine Commission (NYC TLC) provides data on pick-up and drop-off dates/times, pick-up and drop-off locations, trip distances, itemized fares, rate types, payment types, and driver-reported passenger counts for yellow and green taxis available at \url{http://www.nyc.gov/html/tlc/html/about/trip_record_data.shtml}. \newline These taxi trip details are recorded at a resolution of seconds in time and reveal  patterns of urban transportation \cp{vazi:18}.   We  view these data as time-varying networks that are defined by how many people traveled between places of interest and study the evolution of these networks during a typical day. We restrict our analysis to the yellow taxis in the Manhattan area and constructed samples of time-varying networks observed over a sample  of 363 days in the year 2016 for which complete records were available. 

To construct the networks, we identified   69 zones (subareas), which then   form the nodes of the network.  Each day is broken into 5 minute intervals, and for each of these time intervals we constructed a weighted network with these 69 nodes, where   edge weights represent the number of people who traveled between the  pairs of zones (nodes) connected by the corresponding edge within each  5 minute interval.  This generates  a sample of time-varying networks, one time-varying network for each day, sampled over 24 hours for every 45 minutes, where we have 363 such networks, one for   each of thew 363 days in 2016 for which complete records are available. We thus arrive at a sample of 363 time-varying networks.  The observations  at each time point during the 24 hour period correspond to a 69 dimensional graph Laplacian that characterizes the corresponding network between the 69 zones of Manhattan for that particular 5 minute interval. \red{Recently, the  New York taxi data have been analyzed as time-varying networks by various authors, using approaches that substantially differ from the straightforward method proposed here \cp{chen:19,chu:19,mull:19:1}.}

For a network with $r$ nodes, the adjacency matrix is a $r \times r$ matrix $A$ whose $(i,j)^{th}$ entry $a_{ij}$ represents the edge weight between nodes $i$ and $j$, with  corresponding graph Laplacian 
$L=D-A$,
where $D$ is the degree matrix, whose off diagonal entries are zero, with diagonal entries  $d_{ii}=\sum_{j=1}^r a_{ij}$.
We used the Frobenius metric  as a distance measure between graph Laplacians. The sample \F \ mean trajectory at a particular time point then  corresponds to the sample average of the graph Laplacians of 363 networks corresponding to different days for that time point. The movie titled ``mean\_NY.mp4" provided in the Supplement  visualizes  the sample \F \ mean network at different times of the day.  \red{the following goes to the Supplement:  
 with only the top 10\% of edges in increasing order of edge weights for clarity in visualizing the networks.}

We then obtained  the \F \ variance  trajectories for each of the 363 days.  For each fixed time point during the day that indexes one of  the time-varying networks in the sample, we take the squared Frobenius distance between the graph Laplacian of that particular day at the selected  time point and  the \F \ mean graph Laplacian at the same time point. Finally we applied FPCA  to the 363 \F \ variance trajectories. 
The mean \F \ variance function of  the daily graph Laplacians of travel networks between the 69 regions of interest in Manhattan as a function of the time within  the day is shown in Figure  \ref{fig: f5}. \red{it would be useful to show the trajectories themselves, also make the tickmarks larger and omit the title ``Mean Function" in the plot, it is not needed, best omit titles in all plots}  

The predominant directions of variation of the daily \F \ variance trajectories around the \F \ mean function are visualized by the first four eigenfunctions \red{the last one explains less than 5\%, would omit it} in 
Figure \ref{fig: f7}, where the  first four functional principal component scores jointly explain more than 97\% of the variation in the \F \ variance trajectories.
\begin{figure}
	\centering
	\includegraphics[scale = .5]{ny1}
	\caption{Mean function of the squared distance trajectories at 5 minute intervals of graph Laplacians of daily travel networks between 69 zones of Manhattan which is  the \F \ mean graph Laplacian function for the Frobenius metric.}
	\label{fig: f5}
\end{figure}
The first eigenfunction reflects the variability between 12 midnight to 6am, and then variability throughout the day, the second a contrast between early and late morning variability, the third a contrast between noon to evening and morning-late night variability, and the fourth a contrast between morning and late night variability. 

\bco

\begin{figure}
	\centering
	\includegraphics[scale = 0.6]{ny5}
	\caption{Fitted covariance surface for the squared distance trajectories at 5 minute intervals of graph Laplacians of daily travel networks between 69 zones of Manhattan.} 
	\label{fig: f6}
\end{figure}

\fi
\bco
\begin{figure}
	\centering
	\includegraphics[scale = .5]{ny2}
	\caption{Eigenfunctions for the functional principal component analysis of the squared distance trajectories at 5 minute intervals of graph Laplacians of daily travel networks between 69 zones of Manhattan using the Frobenius metric for graph Laplacians. The solid red line  corresponds to the first eigenfunction, which explains 68.32\% of variability in the trajectories, the dashed blue line to the second, which explains 16.86\%, the dotted green line to the third,  which explains 9.26 \% and the brown dot dash line to the fourth, which explains 2.83\%. \red{on all plots make tick marks larger and omit titles} }
	\label{fig: f7}
\end{figure}
\fi
\bco

Analyzing the FPC scores of the daily \F \ variance trajectories along the first and second eigenfunctions, 
Figure \ref{fig: f8} reveals  several interesting patterns in the daily \F \ variance trajectories. Weekdays and weekends form very distinguishable clusters. Holidays  have similar patterns as Sundays. The plots of the FPC scores pinpoint  some  outliers that can be linked with special days. Outliers in the first FPC correspond to New Year's day and November 6, 2016, which is the day when daylight saving ended, and a major outlier in  the second FPC is Christmas day and  in the third FPC it is Independence Day, July 4, 2016.  Using the FPCs, holidays and weekends can be clearly distinguished from working days. An application of a Bayes classifier for functional data \citep{mull:17:6} with the daily \F \ variance trajectories as predictors and ``weekdays" and ``weekends and holidays" as the binary response using the first 200 days of 2016 as the training data and the last 163 days as test data gave a low misclassification rate of 0.0061 when using 4 principal components in the classifier. \red{try using only the first three. Also we should add plots of all trajectories, where the outliers are highlighted ion color}

\begin{figure}
	\centering
	\includegraphics[scale = 0.5]{ny3}
	\caption{Pairwise plots of the FPC scores. Second versus first (top left), third versus first (top right) and fourth versus first (lower left) scores, where   `W' stands for regular Mondays to Thursdays, `Fri' for Fridays, `Sat' for Saturdays, `Sun' for Sundays and `Hol' for special holidays. 
	}
	\label{fig: f8}
\end{figure}

\begin{figure}[t!]
	\centering
	\includegraphics[scale = 0.75]{ny4}
	\caption{First FPC scores corresponding to weekdays  versus the week of the year (left) and second FPC scores corresponding to weekdays  versus the week of the year (right).}
	\label{fig: f9}
\end{figure}

\red{for this perhaps use the eigenanalysis for only M-F and omit Sat-Sun}  In Figure \ref{fig: f9} we plot second against first FPC scores for each of the  work days, Monday to Friday, for the 52 weeks in 2016. We observe that there is some association between the projection scores and the week of the year. For example, for typical work days, during summer months the scores show much less variability as compared to winter months, which is expected as there could be heavy snow or freezing temperatures  that may impact the taxi networks during the winter months more than weather impacts during the summer months.  The plots also reveal several interesting outliers. For Mondays, the outliers correspond to Memorial Day, Independence Day, Labor Day and Boxing Day. For Thursdays, the single outlier corresponds to Thanksgiving and for Friday, the outlier corresponds to New Years Day. 

\bco

\red{\textit{Empirical dynamics of New York taxi data}: Figure \ref{fig: nynew1} illustrates the estimated coefficient of determination function $\hat{R}^2(t)$ and the slope function $\hat{\beta}(t)$ after we implemented the empirical dynamics modeling in \eqref{dyn3} for the sample of time varying taxi trip networks in Manhattan for the New York taxi data. The estimates are different for weekdays and weekends. }
	
\red{For weekdays, $\hat{R}^2(t)$ peaks at 1 am and gradually declines before the morning commute. There are two small peaks, around 9 am and 12:30 pm. It then declines at 4 pm and starts to rise after 7 pm. The slope function is negative between 1 am to 6:30 am and between 11 am to 5 pm, implying that the distance of the network trajectories from the \F \ mean trajectory show dynamic regression to the mean, i.e. they tend to move closer towards their expected behavior. The slope function is positive between 6:30 am to 10:30 am and 6 pm to midnight, which correspond to the periods of morning commute and the evening period. During this time the distance of the network trajectories from the \F \ mean trajectory show explosive behavior, i.e. they tend to move away from their expected behavior. This behavior is expected for morning commute time during weekdays. Several zones in Manhattan have a busy nightlife even during weekdays which attributes for the explosive behavior during this time.
}

\red{During the weekends, $\hat{R}^2(t)$ peaks at 1 am and declines to zero right after 4 pm and starts to rise again after 7 pm. The slope function is negative till before 6 pm, which means that the \F \ variance trajectories tend to move closer to expected behavior during this period. After 6 pm, the slope is positive, indicating that the \F \ variance trajectories have explosive behavior after 6 pm during weekends. It is not surprising that the \F \ variance trajectories tend to move away from their expected behavior during weekend evenings as several zones in Manhattan have a very active nightlife during the weekends. 
}
\begin{figure}
	\centering
	\includegraphics[scale = .6]{ny_new_1}
	\caption{Smooth estimate of the coefficient of determination $R^2(t)$ for weekdays (top left) and weekends (lower left) and the varying coefficient function $\beta(t)$ for weekdays (top right) and weekends (lower right)\citep{mull:10:2} capturing empirical dynamics of the distance trajectories of the New York taxi data.}	\label{fig: nynew1}
\end{figure}

\fi

\subsection{Chicago Divvy Bike Data}

The Chicago Divvy bicycle sharing system makes historical bike trip data publicly available at \url{https://www.divvybikes.com/system-data}. The data set includes trip start and end dates and times, duration, start and end locations and anonymized rider data. The  bike trip details are recorded at a resolution of seconds in time and include trips between 580 bike stations in Chicago and two adjacent suburbs. We used a cleaned version of the  trip records of duration one hour or less between 2013 to 2017, which are  available at \url{https://www.kaggle.com/yingwurenjian/chicago-divvy-bicycle-sharing-data}. \newline 
These data were also analyzed by \ci{gerv:19}, who applied a functional version of  point processes.  Studying patterns in the daily evolution of the number of bike rides between various bike stations can provide insights into the Divvy bike sharing system and patterns of transport in the city. We study time-varying networks that are defined by the number of bike trips between stations of interest and the evolution of these
networks during a typical day. We constructed samples of time-varying networks observed over a sample of 1457 days in the years 2013 to 2017. 

We focus our analysis on data pertaining to  the area east of Greektown, south of Wrigley field and north of Chinatown containing  the Lakefront trails, the Navy pier and many other popular destinations. We considered 112  popular bike stations in this region and each day was  broken into 20 minute intervals. On each of these intervals we constructed a network with 112 nodes, each one corresponding to one of the bike stations and edge weights representing the number of recorded  bike trips between the pairs of stations that define  the edges of the network  within the 20 minute interval.  This generates a time-varying network for  each of the 1457 days in the years 2013 to 2017 for which complete records are available. 
The time points where the network is sampled over the course of each day were chosen as the midpoints of  the 20 minute intervals of a day. The  observations at each time
point correspond to a 112 dimensional graph Laplacian that characterizes the  network between the 112 bike stations of interest for that particular 20 minute interval. For a
network with $r$ nodes, the adjacency matrix is a $r \times r$ matrix $A$, where the $(i,j)^{th}$ entry $a_{ij}$ represents the edge weight between nodes $i$ and $j$. The graph Laplacians $L$ are given by $L=D-A$, where $D$ is the degree matrix, the off diagonal entries of which are zero, with diagonal entries $d_{ii}=\sum_{j=1}^{r} a_{ij}$. The graph Laplacians determine the network uniquely. 

We used the Frobenius metric as a distance measure between graph Laplacians. The sample \F \ mean trajectory at a particular time point therefore is the sample average of the graph Laplacians of 1457 networks corresponding to different days
for that time point. We then obtained the \F \ variance trajectories for each day, which for a given day and  time point correspond to the squared Frobenius distance between the graph Laplacian  and  the \F \ mean graph Laplacian,  and then applied functional principal component analysis (FPCA)   for the resulting 1457 \F \ variance trajectories. 

The mean \F \ variance trajectory of the daily graph Laplacians for the Divvy bike trip networks as a 
function of the time within the day, which quantifies the average squared deviation from the mean trajectory,  is shown in the left plot of Figure \ref{fig: b1}. The peaks are at 9am with elevated mean variation between between 7am to 10am and at 6pm with elevated levels between 4pm to 7pm,  which reflect morning and late afternoon and early evening  commuting surges, where the network variation is seen to be highest. 

\begin{figure*}[t!]
	\centering
	\hspace{-10mm}
	\begin{subfigure}[t]{0.5\textwidth}
		\centering
		\includegraphics[width=0.9\textwidth]{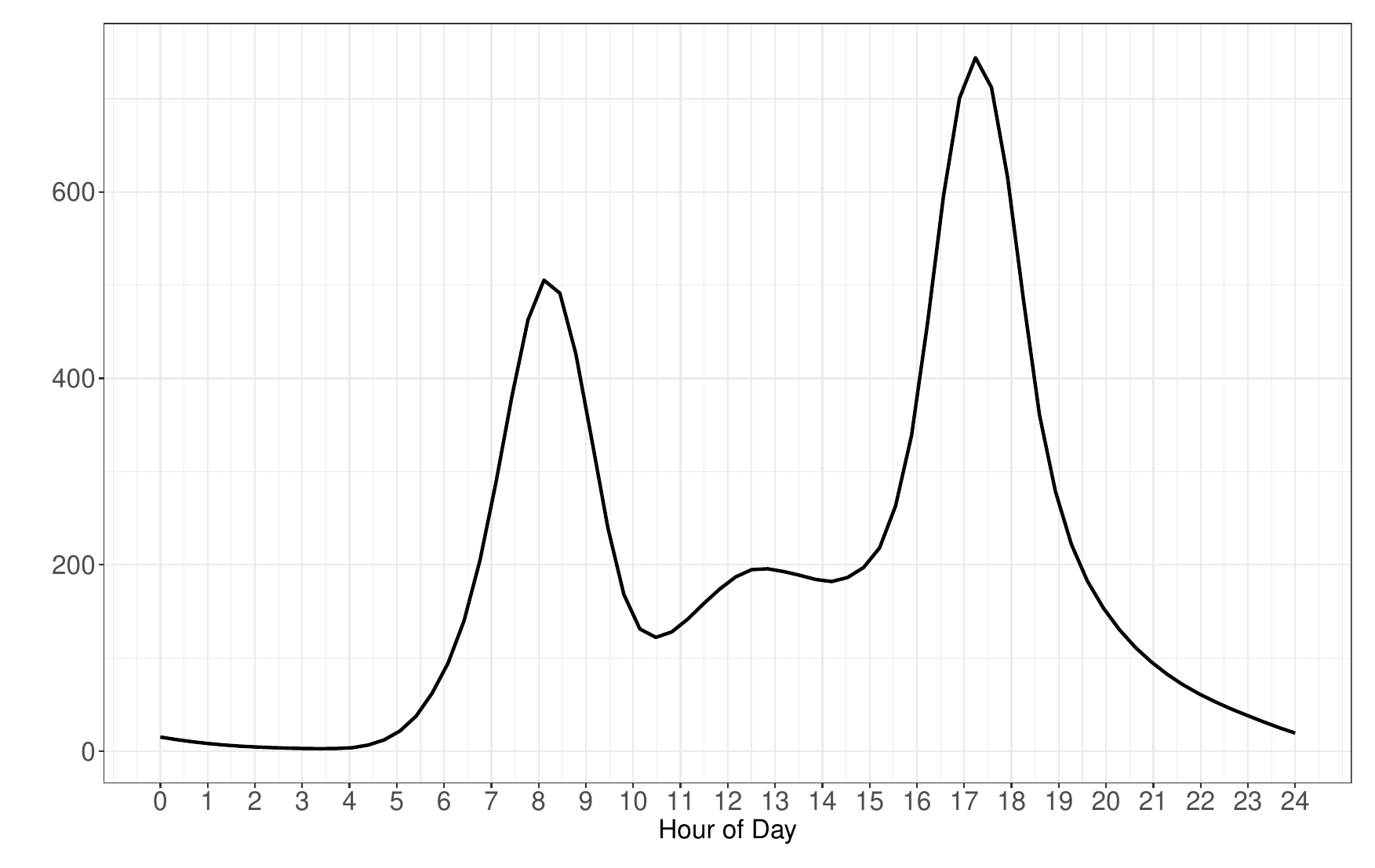}
	\end{subfigure}%
    ~
	\hspace{-10mm}
	\begin{subfigure}[t]{0.5\textwidth}
		\centering
		\includegraphics[width=0.9\textwidth]{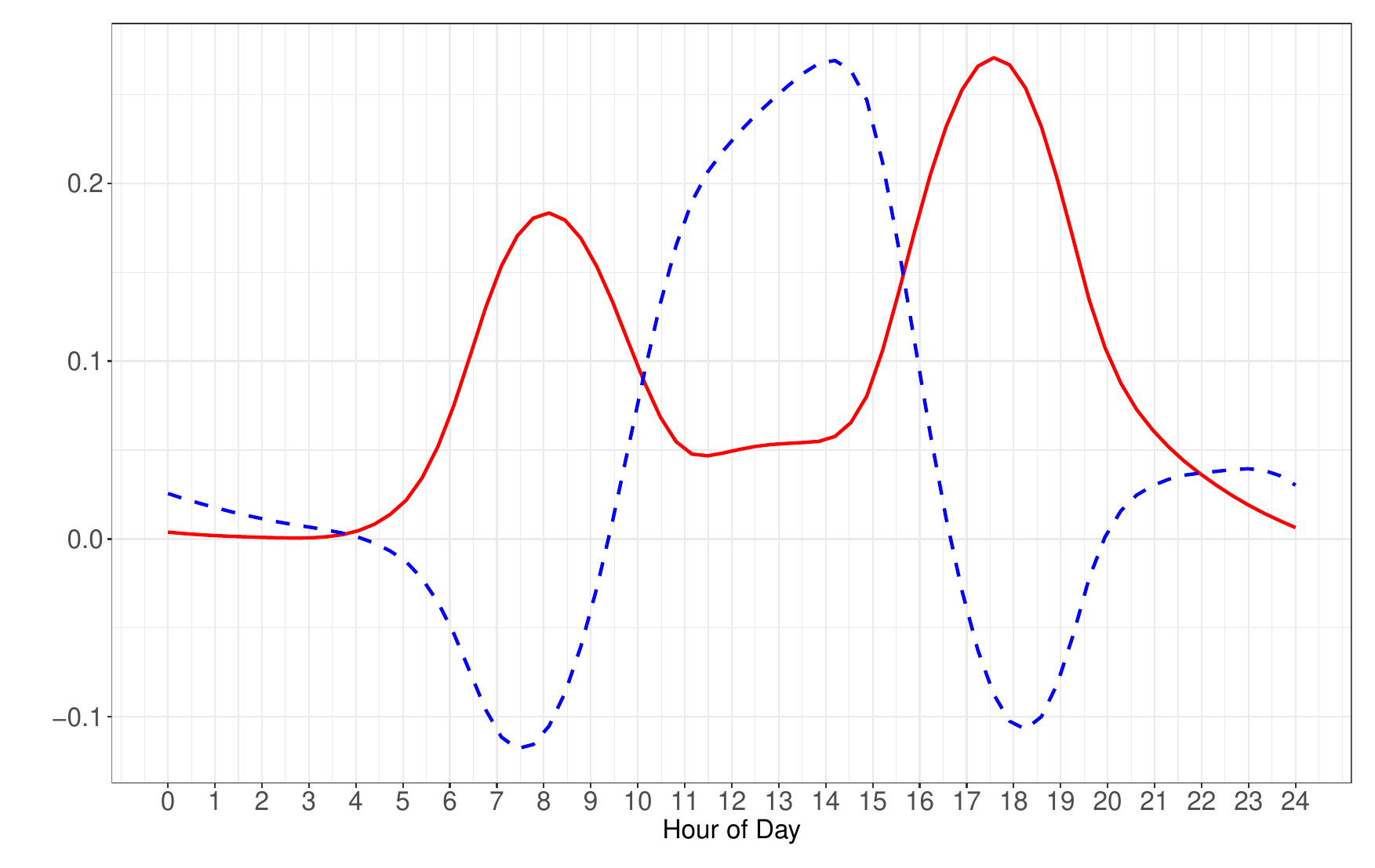}
	\end{subfigure}
	\caption{Sample mean function (left plot) and eigenfunctions for the functional principal component analysis (right plot) of the squared distance trajectories at 20 minute intervals of graph Laplacians of daily Divvy bike trip networks in Chicago. In the right plot, the solid red line corresponds to the first eigenfunction, which explains 90.40\% of variability in the trajectories and the dashed blue line to the second eigenfunction, which explains 7.28\% of the variability. }
	\label{fig: b1}
\end{figure*}

The predominant directions of variation of the daily \F \ variance trajectories around the \F \ mean function are visualized by the first two eigenfunctions in the right plot of Figure \ref{fig: b1}. The first two functional principal component scores explain about 97.7\% of the
variation in the \F \ variance trajectories. The first eigenfunction reflects increased variability around the peaks of the \F \ variance function that is shown in Figure \ref{fig: b1}. The peaks are between between 7am to 10am and 4pm to 7pm,  which reflect morning and late afternoon and early evening peaks of commute, where the deviations from the mean \F  \ variance  function  are seen to be largest. The second eigenfunction  reflects  a contrast between these peaks and the squared deviation from the mean \F \ variance function during the time period 11am to 4pm. 

\begin{figure}
	\centering
	\includegraphics[scale = .6]{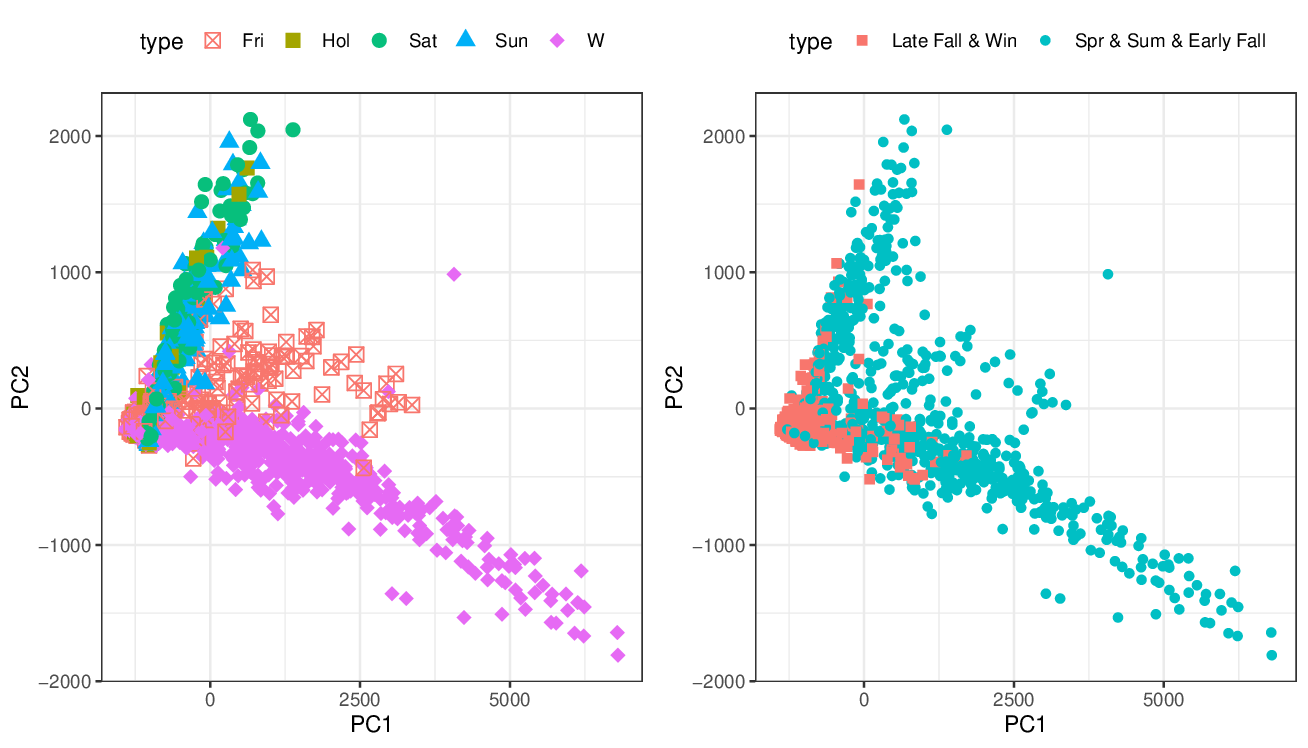}
	\caption{Pairwise plots of the first two FPC scores, distinguished  by day of the week (left plot) and by season (right plot). In the left plot, `W' stands for regular Mondays to Thursdays, `Fri' for Fridays, `Sat' for Saturdays, `Sun' for Sundays and `Hol' for special holidays. In the right plot, `Late Fall \& Win' includes months between November and March and `Spr \& Sum \& Early Fall' includes April to October. } 
	\label{fig: b3}
\end{figure}

Analyzing the FPC scores of the daily \F \ variance trajectories along the first and second eigenfunctions, Figure \ref{fig: b3} reveals several interesting patterns in the daily \F \ variance trajectories. Weekdays and weekends form distinct clusters. 
Holidays show  patterns similar to  weekends. When scrutinizing  second versus first FPC scores, an outlying observation is
found at  August 21, 2017. Researching  the background of this day, we found that  there was a total solar eclipse, which was in peak view over Chicago at 1:18 pm in the afternoon. The \F \ variance trajectory of this particular day is illustrated in Figure \ref{fig: s1} in section \hyperref[supp4]{A.4} of the online supplement. An application of the naive Bayes classifier and the support vector machine on the first two FPC scores of the daily \F \ variance trajectories with  ``weekdays" and ``weekends and holidays" as the binary response using 75\% of the data as training sample and 25\% of the data as test sample gave a low misclassification rate of 6.02\% in both cases. The classification result is illustrated in Figure \ref{fig: s2} in the online supplement. 

We also performed a  FPCA  of these \F \ variance trajectories for weekdays, Fridays and weekends, including special holidays in the same group as weekends, separately for the three cases, with results displayed  in Figures \ref{fig: s3} and \ref{fig: s4} in  the online supplement. Weekdays and weekends, including special holidays, show clear differences in  both mean \F \ variance functions and eigenfunctions. The Friday pattern can be characterized as  ``transition"  from weekdays to weekends.
Seasonal differences can impact bike sharing patterns. In the right plot of Figure \ref{fig: b3}, we display second versus first FPC scores, differentiated according to two broad seasonal groups. Spring, summer and early fall includes months from April to October that exhibit greater variability than the late fall and winter months of November to March, with further illustrations  in Figures \ref{fig: s5} and \ref{fig: s6} in the online supplement.

\subsection{Fertility Data}
The Human Fertility Database provides cohort fertility data for various countries and calendar years. The data are available at \url{www.humanfertility.org} and facilitate the study  of the time evolution and inter-country differences in fertility over a period spanning  more than 30 calendar years \citep[see also][]{mull:17:4}.  
We selected 27 countries with complete fertility records for the time period 1976 to 2009. For each country and year, the age specific total live birth counts correspond to histograms of maternal age with bin size one year. These histograms were smoothed (for which we employed  local least squares smoothing using the Hades package available at \url{https://stat.ucdavis.edu/hades/})  to obtain smooth probability density functions for maternal age, where we consider the age interval $[12,55]$. We thus  obtain samples of time-varying univariate probability distributions, where the subjects are the countries, the time is  calendar years between 1976 and 2009 and the observation at each time point for a specific country is its maternal age distribution, i.e., the distribution of ages when females give birth within the age  interval $[12,55]$ for the specified country and calendar year.  

\begin{figure}[H]
	\centering
	\includegraphics[width=\textwidth]{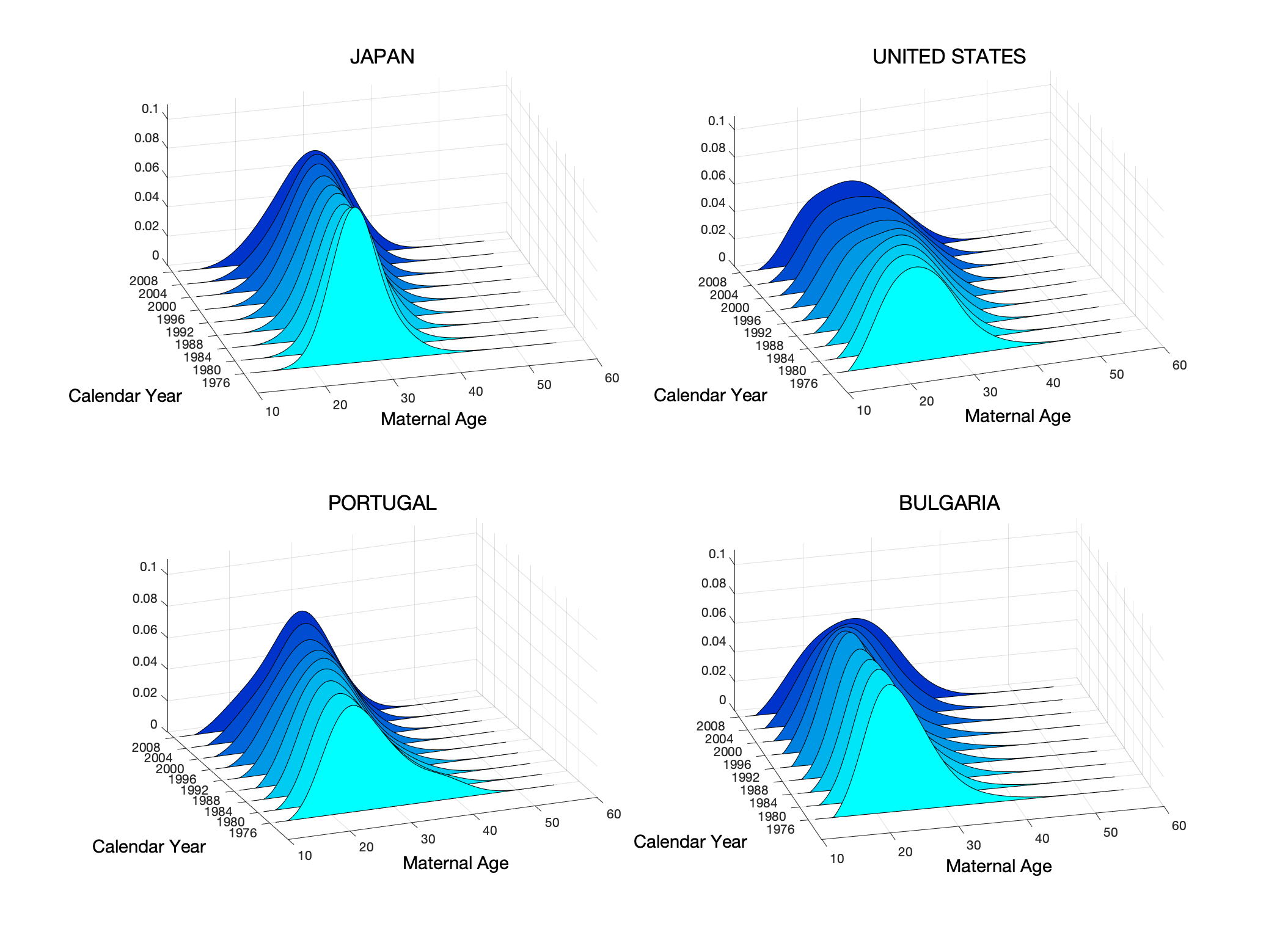}
	\caption{Yearly maternal age distributions represented as density functions for the age interval [12,55] during the time period 1976 to 2009 (selected years). The $x$-axis represents maternal age, the $z$-axis the density functions and the $y$-axis  calendar years.}
	\label{fig: f20}
\end{figure}

Figure \ref{fig: f20} displays  the evolving densities for  Japan, United States, Portugal and Bulgaria in some selected years over a period spanning 56 years from 1959 to 2014. There are clear differences in the maternal age evolution between countries, but the overall trend is that maternal age increases. This is also reflected in the  \F \ mean densities in Figure \ref{fig: f3}, which show a shift in their mode locations towards  higher age over the years. 

 We opt for the 2-Wasserstein metric as the distance between probability distributions, which corresponds to the  $L^2$ distance between their quantile functions, a metric that has proved to be an excellent choice in many applications \cp{bols:03,bigo:18}.  The sample \F \ mean trajectory in a particular year then corresponds to the sample average of the quantile functions of the 27 countries in that year \citep{mull:18:5}, which is represented as the corresponding density function. We then obtained  the \F\ variance trajectory for each country, where its value for a given calendar year corresponds to the squared 2-Wasserstein distance between the maternal age distribution of the country  to the \F \ mean age distribution for the specified calendar year.  
 
 Finally,  we performed a FPCA  on the 27 country specific \F \ variance trajectories. Figure \ref{fig: f1} shows the estimated mean function of the \F\ variance trajectories, which is a scalar function that corresponds to the sample average of the \F \ variance trajectories as a function of calendar year. The predominant directions of variation of the distance trajectories around the \F \ mean trajectory are captured by the first two eigenfunctions, which are depicted in the left plot of Figure \ref{fig: f2}.  The first and second eigencomponent  explain 74.64\% and 17.40\% of the variation in the distance trajectories. The first eigenfunction is increasing until between 1995 to 2000 and then starts to decrease. This shows increasing deviation of the distance trajectories  from the mean \F \   variance trajectory until right before the new millennium after which these deviations tend to decrease. The second eigenfunction reflects  a contrast between the earlier and later parts of the calendar time interval.
 
 \begin{figure}
	\centering
	\includegraphics[scale = .4]{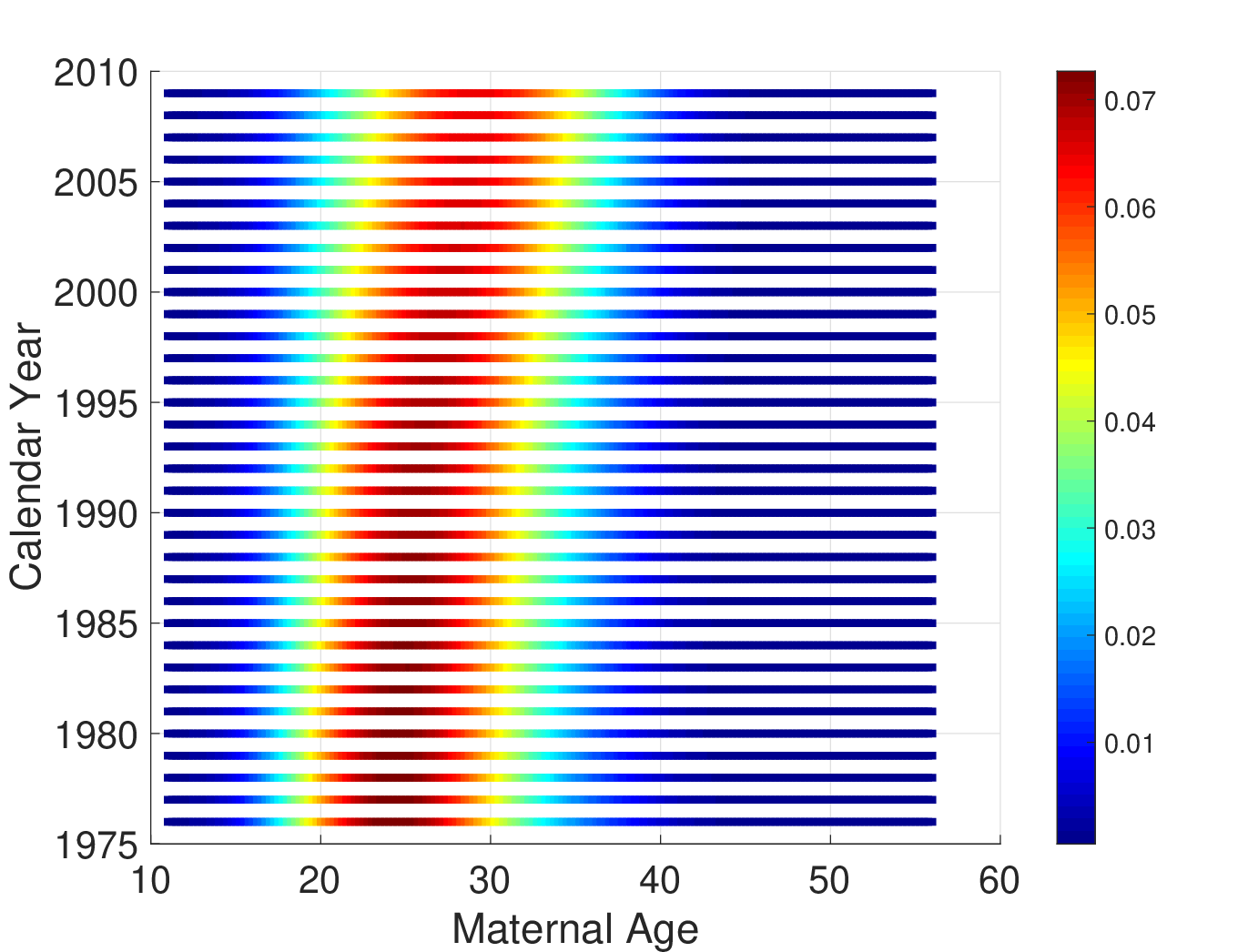}
	\caption{\F \ mean maternal age distributions represented as a heat plot of density functions for the age interval [12,55] during the time period 1976 to 2009. }
	\label{fig: f3}
\end{figure}
 The years between 1995 to 2000 mark the critical period when the changes in the mode of maternal age distributions start to take place as displayed in Figure \ref{fig: f3}. This period of increased activity is also prominent in Figure \ref{fig: f1} which shows that the \F \ variance of maternal age distributions increases in the beginning, reaches a peak  between 1995 and 2000 and then starts to decrease. This might be attributed to increasing numbers of women opting for higher education and participating in the labor force in some of the countries included in the dataset over the time interval where the data have been collected. Another likely factor are advanced birth control measures in the past few decades, which led to changes in the  maternal age distribution early on for some countries, while these changes were delayed for some other countries, leading to increased discrepancies between countries from  1995 to 2000,  which stabilized  later as  countries moved closer to the mean behavior.   
\begin{figure} 
	\centering
	\includegraphics[width=0.4\textwidth]{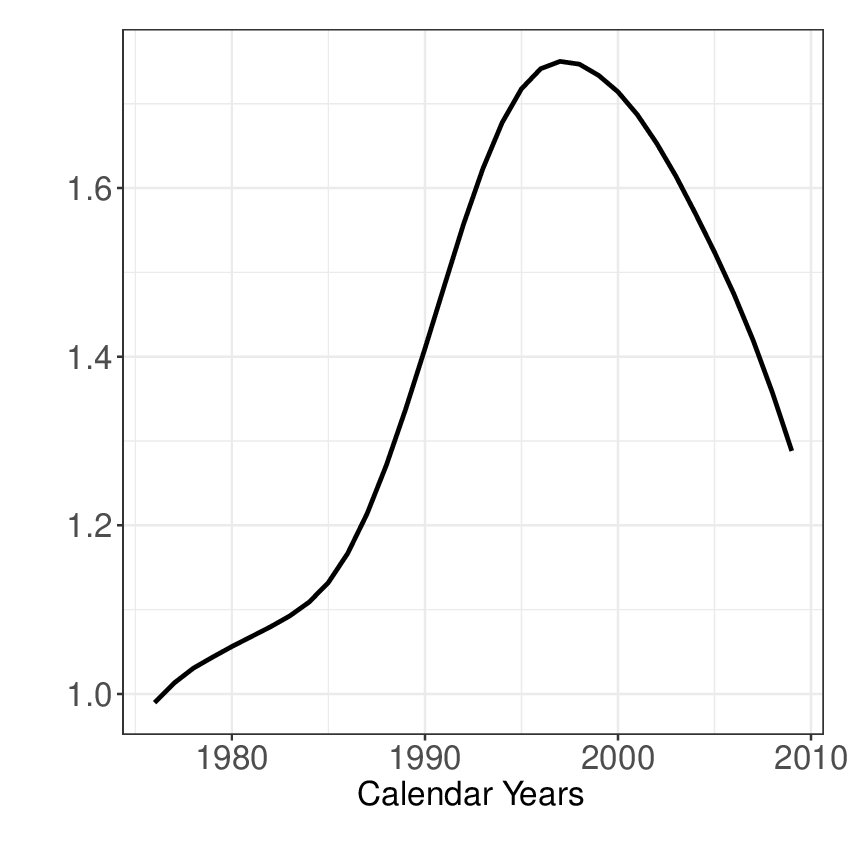}
	\caption{Sample mean  function 
	of the \F \ variance trajectories of the time courses of country specific maternal age distributions.}  
	\label{fig: f1}
\end{figure}

\bco

\begin{figure}
	\centering
	\includegraphics[scale = .65]{fert4}
	\caption{Fitted smoothed covariance surfaces of the \F \ variance trajectories of country specific maternal age distribution time courses from their \F \ mean trajectory.}
	\label{fig: f21}
\end{figure}
 
 \fi

\begin{figure*}[t!]
	\centering
	\hspace{-10mm}
	\begin{subfigure}[t]{0.5\textwidth}
		\centering
		\includegraphics[scale=0.5]{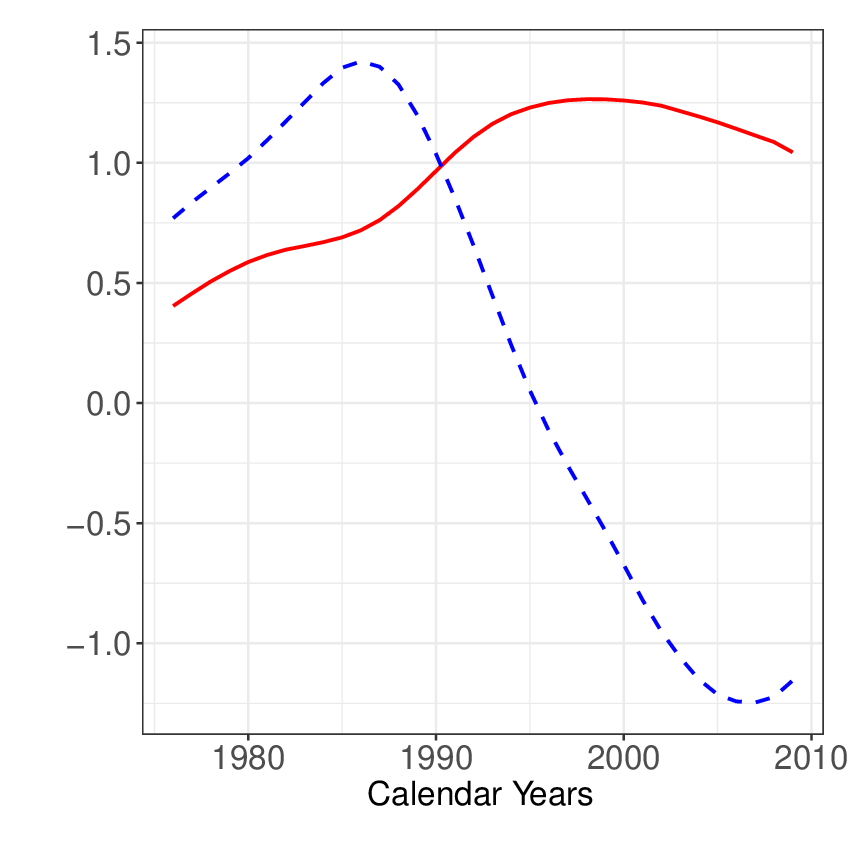}
	\end{subfigure}%
	~
	\begin{subfigure}[t]{0.5\textwidth}
		\centering
		\includegraphics[scale=0.5]{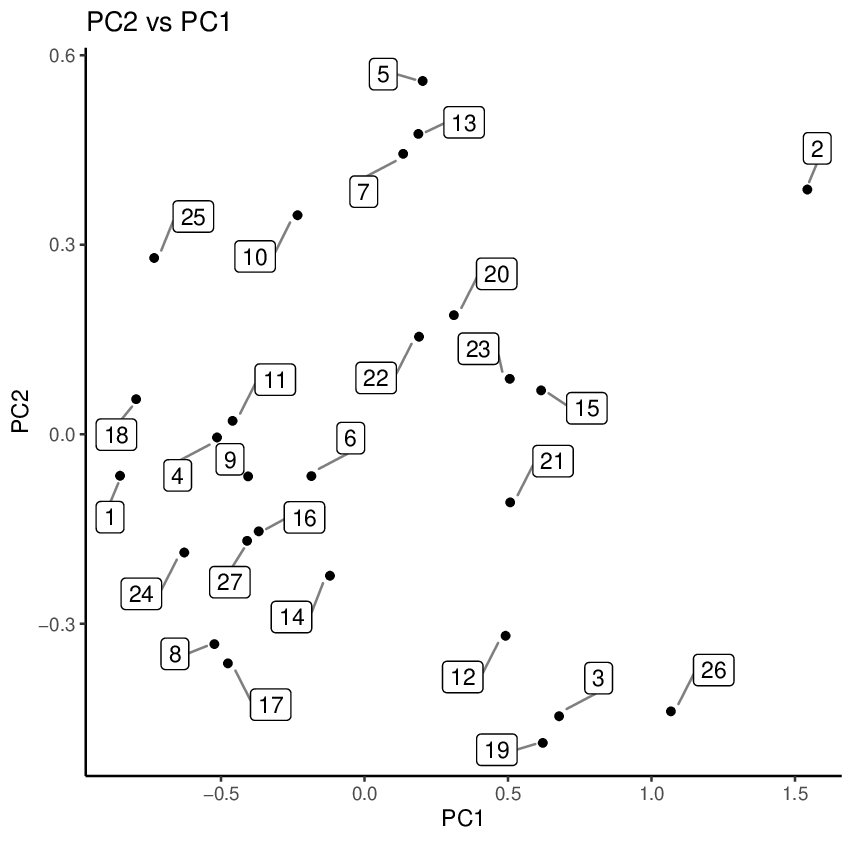}
	\end{subfigure}
	\caption{Left panel: First (solid red)  and second (blue dashed) eigenfunction  of the country specific \F \ variance trajectories, which  explain 74.64\% and 17.40\% of variability,  respectively. Right panel: Second versus first FPC scores of the country specific \F \ variance trajectories.  Here   1  denotes Austria; 
		2  Bulgaria; 
		3  Belarus; 
		4  Canada; 
		5  Czech Republic; 
		6  Estoina; 
		7  Finland; 
		8  Germany; 
		9  France; 
		10  Hungary; 
		11  Israel; 
		12  Italy; 
		13  Japan; 
		14  Lithuania; 
		15  Netherlands; 
		16  Norway; 
		17  Poland; 
		18  Portugal; 
		19  Russia; 
		20  Slovakia; 
		21  Spain; 
		22  Sweden; 
		23  Switzerland; 
		24  United Kingdom; 
		25  Taiwan; 
		26  Ukraine; 
		27  United States.}
	\label{fig: f2}
\end{figure*}

The FPC scores of the trajectories along predominant directions of variation not only reveal interesting patterns but also aid in identifying extremes, which is tricky for the case of time-varying probability distributions. Plotting second versus first FPC scores (right plot of Figure \ref{fig: f2}), one finds that 
Bulgaria shows large deviations from the \F \ mean variance function trajectory along both first and second eigenfunctions.  The Czech Republic has the highest second FPC, indicating a  large contrast in the distance from the \F \ mean trajectory between early and later years. Portugal and Austria have negative first FPCs and small second FPCs,  which means that their deviation from the \F \ mean trajectory  is less than the average deviation over the calendar time interval.

\subsection{Empirical Dynamics}

When using FDA for analyzing real valued longitudinal data, a common assumption is that the data are generated by an underlying smooth and square integrable stochastic process. The derivatives of the trajectories of such processes are  often the key to understanding the underlying dynamics  \cp{rams:00, mas:09, rams:07}. Empirical dynamics 
\citep{mull:10:2} is a systematic approach to assess the underlying dynamics of longitudinally observed data. It is based on the decomposition 
\begin{equation}
\label{dyn}
Y^{(1)}(t)-\mu^{(1)}(t)= \beta(t) \{Y(t)-\mu(t)\}+ Z(t),
\end{equation}
where $Y(\cdot)$ is the underlying stochastic process, $\mu(t)=E(Y(t))$, $Y^{(1)}(\cdot)$ and $\mu^{(1)}(\cdot)$ are the derivatives of $Y(\cdot)$ and $\mu(\cdot)$, $\beta(\cdot)$ is a smooth time varying linear coefficient function and $Z(\cdot)$ is a random drift process. For Gaussian processes this decomposition can be easily derived and  then $Z$ and $Y$ are independent with  $E(Z(t))=0$, leading to 
\begin{equation}
\label{dyn2}
E\left[Y^{(1)}(t)-\mu^{(1)}(t) \Big| Y(t) \right]= \beta(t) \{Y(t)-\mu(t)\}.
\end{equation} 

When one has non-Gaussian processes, such as the  distance processes $V^\ast$ (\ref{or}),  the above decompositions provide useful approximations and can be interpreted in a least squares sense, where the 
coefficient of determination
\begin{equation*}
R^2(t)= 1-\frac{\var\{Z(t)\}}{\var\{Y^{(1)}(t)\}}
\end{equation*} indicates the fraction of variance explained by the empirical dynamics approximation.
On pertinent subdomains where $R^2(t)$ is relatively large,  the dynamics of the trajectories are determined to a greater extent by the linear model in \eqref{dyn2}, while otherwise the random drift process $Z$ becomes the driving factor instead of the dynamic equation. The time-varying function $\beta(t)$  summarizes the characteristics of the dynamics of the underlying process.  If $\beta(t) <0$, one observes {\it centripetality} or {\it dynamic regression to the mean},  i.e., a trajectory which is away from the mean function tends to move closer toward the mean function as time progresses. If on the other hand $\beta(t)>0$, one has {\it centrifugality} or {\it dynamic explosive behavior},  since deviations from the mean at time $t$ tend to increase beyond $t$.

Noting that it is infeasible to define derivatives for trajectories of random objects since the notion of derivative for metric space data is not defined, we apply empirical dynamics instead to the real-valued   \F \ variance trajectories $V^\ast$ (\ref{or}), 
	\begin{align}
	\label{dyn3}
      \frac{d}{dt} (\sv-E\sv) =  \beta(t) (\sv-E\sv)+Z(t).
	\end{align}
We note that centripetality and centrifugality relate to the mean trajectory $E\sv$, i.e., the mean \F \ variance function;  
we quantify the dynamics of the variation relative to this  function.  For estimation of the time varying coefficient $\beta(\cdot)$ and $R^2(t)$, we adopt a  plug-in estimation procedure  \citep{mull:09:1,mull:10:6}.

We implemented the empirical dynamics model \eqref{dyn3} for the sample of time varying maternal age distributions 
with  the R package \textit{fdapace}.  Figure \ref{fig: fnew1} illustrates the estimated slope function $\hat{\beta}(t)$ and coefficient of determination function $\hat{R}^2(t)$, where  the latter indicates that the  dynamics of the maternal age distribution trajectories can be explained by the first order differential equation in the period before 1995 and after 2000, but not in between.  The slope function is positive before 2000 and negative after 2000, indicating centrifugality of the maternal age distributions before 2000 and centripetality after 2000, so that in  more recent years there is a tendency for the fertility distributions to become more similar in the sense that their variation around the mean distance function is decreasing over calendar time, while the mean distance function itself is also decreasing as seen in  Figure \ref{fig: f1}.

\begin{figure*}[t!]
	\centering
	\begin{subfigure}[t]{0.5\textwidth}
		\centering
		\includegraphics[width=\textwidth]{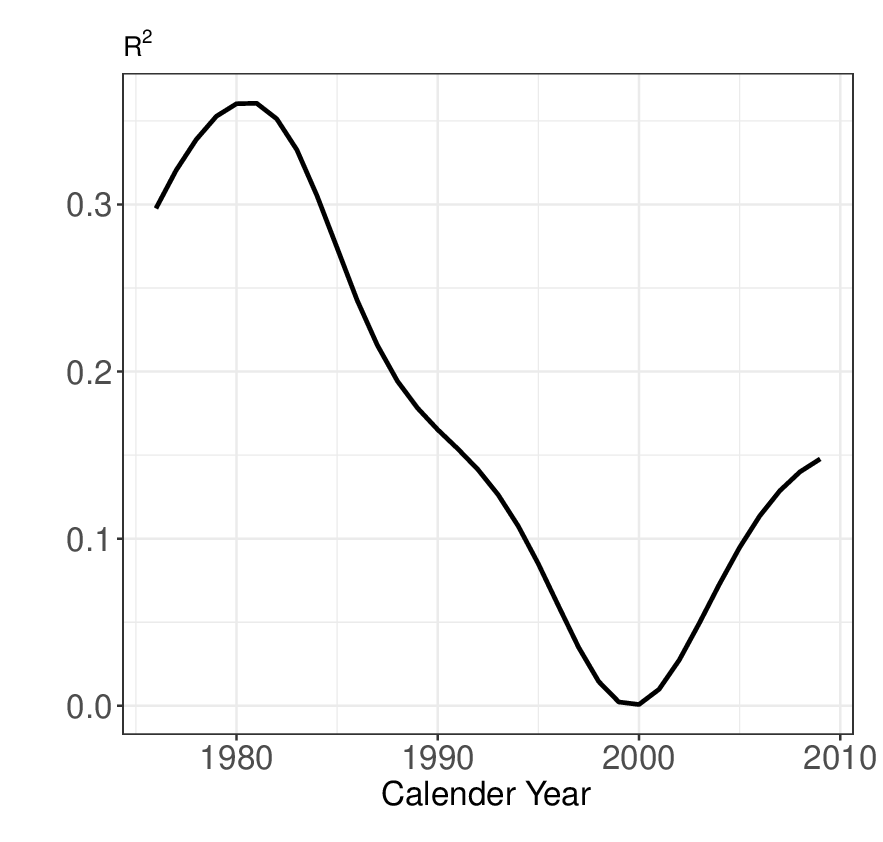}
	\end{subfigure}%
	~ 
	\begin{subfigure}[t]{0.5\textwidth}
		\centering
		\includegraphics[width=\textwidth]{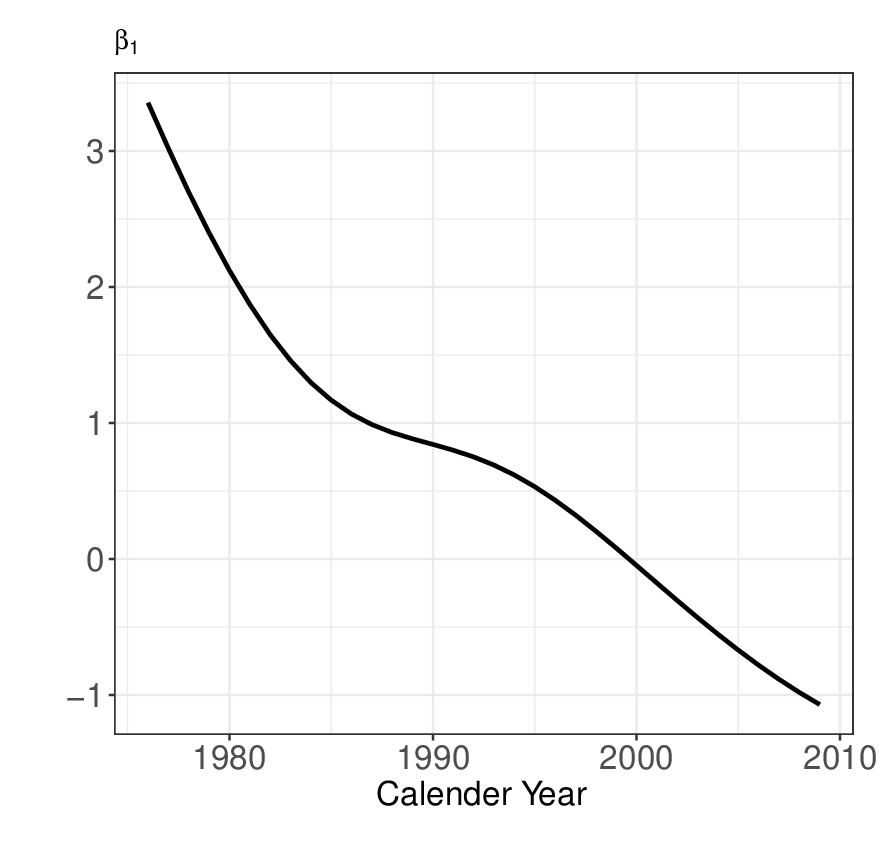}
	\end{subfigure}
	\caption{Smooth estimate of the coefficient of determination $R^2(t)$ (left) and the varying coefficient function $\beta(t)$(right),  capturing empirical dynamics of the distance trajectories of the fertility data.}
\label{fig: fnew1}
\end{figure*}

\subsection{Z\"urich Longitudinal Growth Data}

Statistical shape analysis is an emerging field of data analysis that constitutes measuring, describing and comparing random shapes \citep{dryd:16,kend:84,kend:89,le:93,patr:15,mard:05}. Often, shapes are determined using a finite set of coordinate points, known as landmark points. As an example of shapes evolving over time, we consider the growth modalities that are obtained from a longitudinal study on human growth and development \citep{mull:84:2}. This study included growth data for $232$ Swiss children and was conducted in the University Children’s Hospital in Z\"urich between 1954 and 1978. For each child, we consider trajectories $X_i(t)$, where $X_i(t)$ is a $4 \times 2$ matrix which represents four landmarks in two dimensions. The four landmarks are the foot, which is set to be the point $(0,0)$, the top of the head which has coordinate $(0,\text{standing height})$, the shoulder which is set to have the coordinate $(\text{bi-humeral diameter}, 0.8 \ \text{standing height})$ and the hip which is set at the point $(\text{bi-iliacal diameter},\text{leg length})$.  These shape trajectories are observed for the age interval $t \in [0.5, 20]$ years. 

The trajectories $X_i(t)$ described above correspond to two-dimensional landmarks that take values in the planar shape space, which 
can be viewed as configurations on the complex plane. In this space, we adopt the full Procrustes metric, which  has been quite successful in the analysis of planar shapes in a wide variety of applications \citep{dryd:16}. Given centered complex configurations $y=(y_1,y_2, \dots ,y_k)$ and $z=(z_1,z_2,\dots,z_k)$, both in $\mathbb{C}^k$, with $y^* 1_k=0=z^* 1_k$, the full Procrustes distance between $y$ and $z$ is defined as
\begin{equation*}
d_{\mathcal{P}}(y,z) = \left \lbrace 1- \frac{y^*zz^*y}{z^*zy^*y} \right \rbrace^{1/2}.
\end{equation*}

Accordingly, we first obtained  the sample \F \ mean trajectory of the random shape trajectories under the full Procrustes metric as
\begin{equation*}
\hat{\mu}(t) = \argmin_{\omega \in \Omega} \frac{1}{n} \sum_{i=1}^n d^2_{\mathcal{P}}(X_i(t),\omega),
\end{equation*}
and then analyzed the subject specific \F \ variance trajectories $d^2(X_i(t),\hat{\mu}(t)), t \in [0.5,20], i=1,2, \dots,232$ using the tools developed in this paper. The sample \F \ mean and the computation of the full Procrustes distance were carried out using the R package \texttt{shapes} \citep{dryd:12}. {For details on existence and uniqueness of \F \ means in shape spaces see \cite{le:95,le:98}}.

\begin{figure} 
	\centering
	\includegraphics[width=0.5\textwidth]{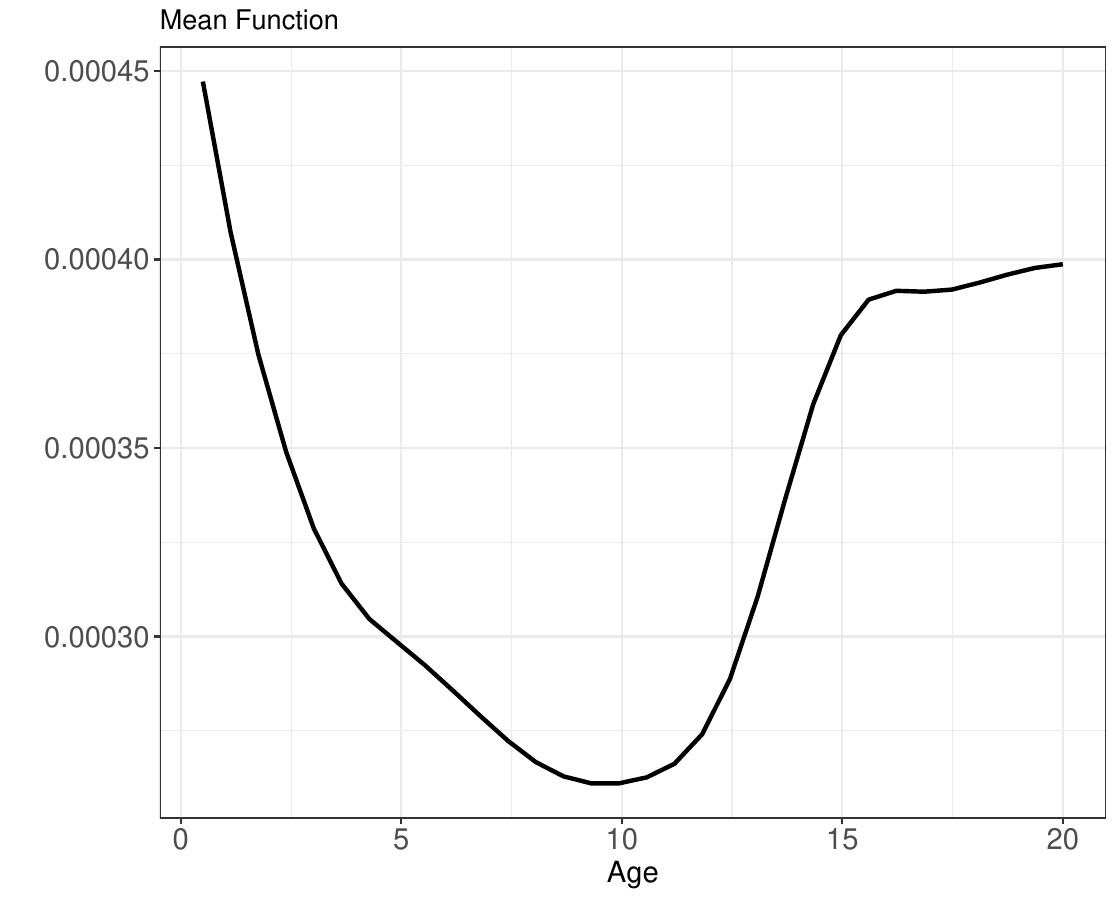}
	\caption{Sample mean  function 
	of the \F \ variance trajectories of the time courses of shape data representing growth modalities.}  
	\label{fig: s1}
\end{figure}

\begin{figure*}[t!]
	\centering
	\hspace{-10mm}
	\begin{subfigure}[t]{0.5\textwidth}
		\centering
		\includegraphics[width=\textwidth]{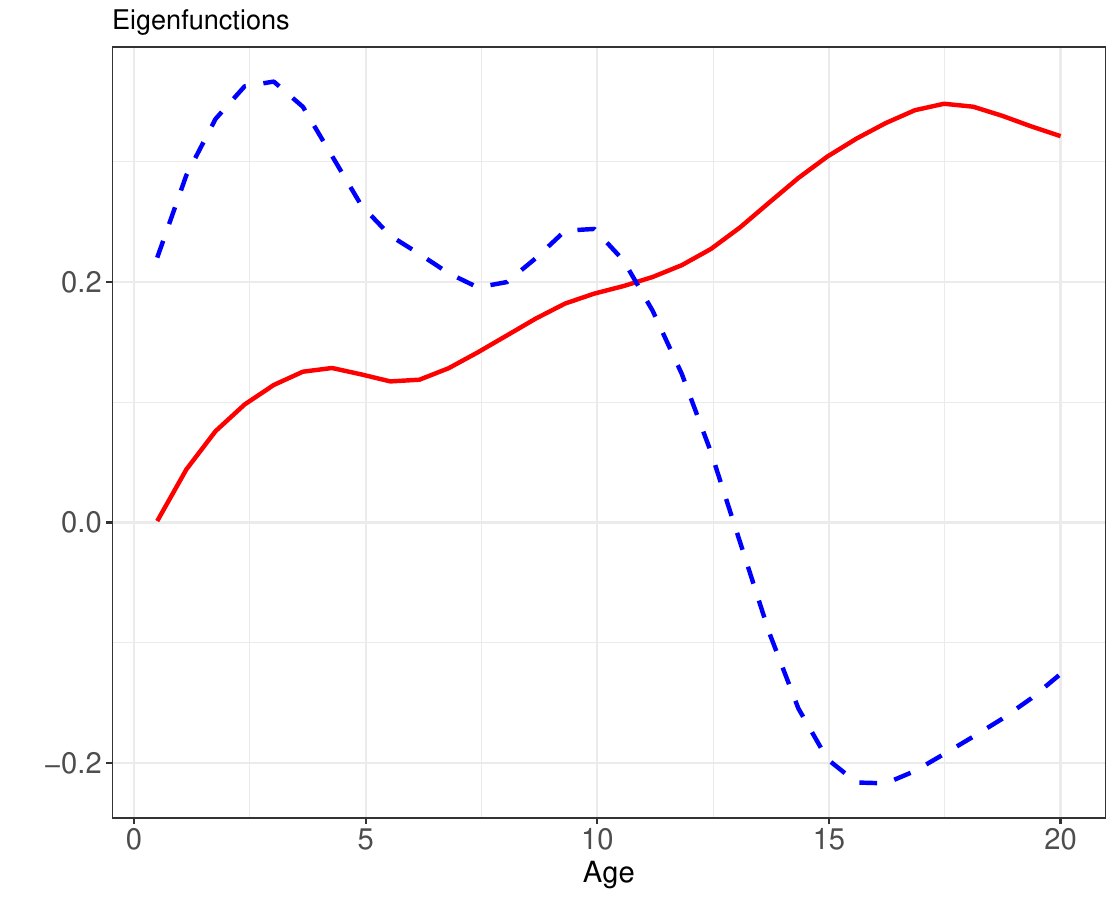}
	\end{subfigure}%
	~
	\begin{subfigure}[t]{0.5\textwidth}
		\centering
		\includegraphics[width=\textwidth]{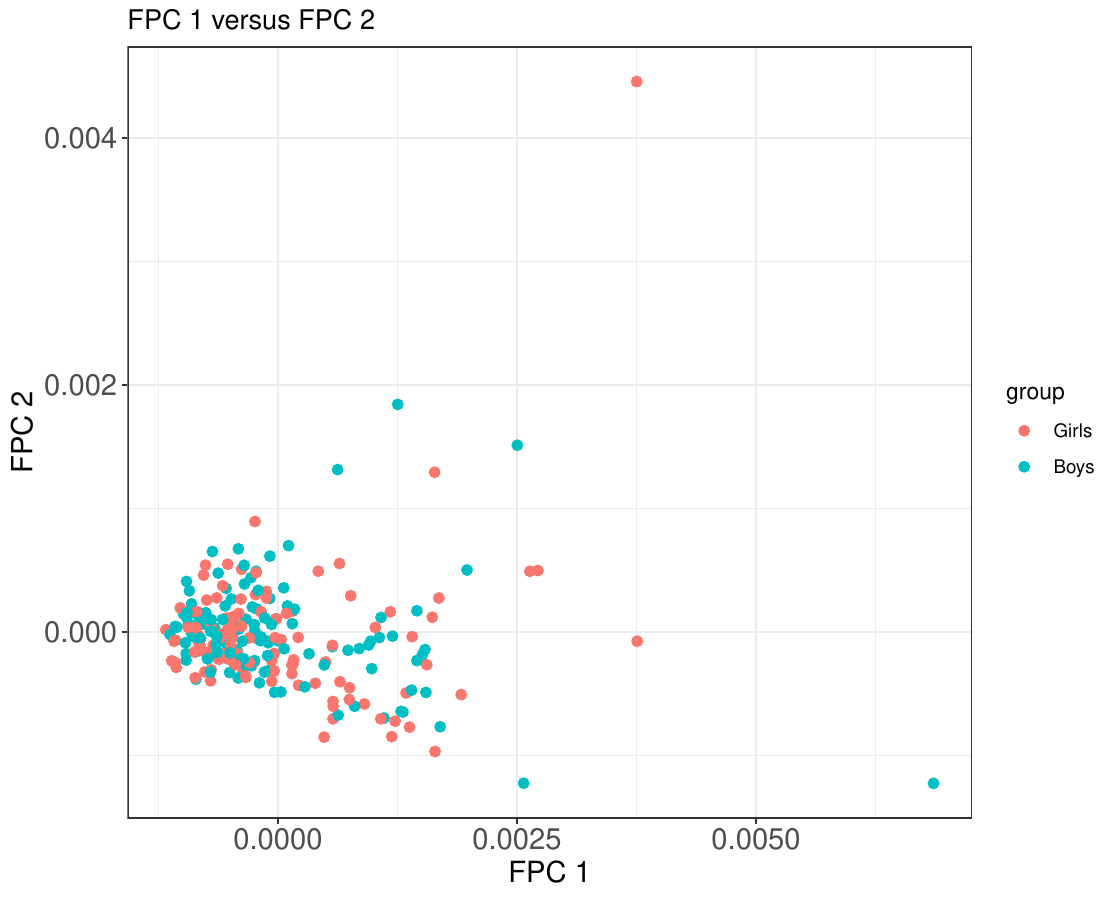}
	\end{subfigure}
	\caption{Left panel: First (solid red)  and second (blue dashed) eigenfunction  of the child specific \F \ variance trajectories, which  explain 76.07\% and 18.55\% of variability,  respectively. Right panel: Second versus first FPC scores of the child specific \F \ variance trajectories. }
	\label{fig: s2}
\end{figure*}

Figure \ref{fig: s1} shows the time evolution of population \F \ variance, which captures the overall variability trends of the shape trajectories around the \F \ mean trajectory. The periods of early childhood, that is from infancy to about 5 years of age, and the period starting at adolescence until adulthood, i.e. between 13 to 20 years, exhibit greater variability around the \F \ mean shapes. Figure \ref{fig: s2} show the first two dominant eigenfunctions, which explain 76.07\% and 18.55\% of the variability in the subject specific \F \ variance trajectories. The corresponding plot of the second  FPC score versus the first  FPC score shows that there is no systematic  difference between boys and girls  when comparing the shape trajectories under the full Procrustes metric. 

{To illustrate the growth shape patterns, we select four boys as  indicated   in the left panel of Figure \ref{fig: s3} with the first FPC scores ranging from positive to zero to negative.  A negative first FPC score of the subject specific \F \ variance trajectories loads negatively on the first eigenfunction illustrated in Figure \ref{fig: s2}, which highlights growth patterns that are closer to the sample \F \ mean shape trajectory than average. A first FPC score closer to zero is associated with subjects that exhibit close to average variability around the sample \F \ mean shape over time, while  a positive first FPC score suggests greater than average variability around the mean \F \ mean trajectory. The four boys ``1",``2", ``3" and ``4" illustrated in Figure \ref{fig: s3} fit these prototypes. While ``1" is closest to the sample \F \ mean shape trajectory, ``2" represents a child who shows typical deviation around the \F \ mean shape over the years, whereas ``3" exhibits larger than typical variation over the years and ``4" stands out from the rest with an extremely high first FPC score. This pattern suggests that the first FPC captures differences in overall size  
of the children.} 

{Figure \ref{fig: s4} illustrates variability with respect to the second eigenfunction, which   corresponds to a contrast  between early and later years 
(Figure \ref{fig: s2}).   
for three selected boys highlighted in the left panel, all of whom have small first FPC scores, and whose second FPC scores vary from negative to zero to positive. For these boys, `5" tends to remain thinner during the teenage years as compared to``6" whose hip diameter gets wider during the teenage years, while ``7" has the biggest growth in hip diameter among the three. }

\begin{figure*}[t!]
	\centering
	\hspace{-10mm}
	\begin{subfigure}[t]{0.3\textwidth}
		\centering
		\includegraphics[width=\textwidth]{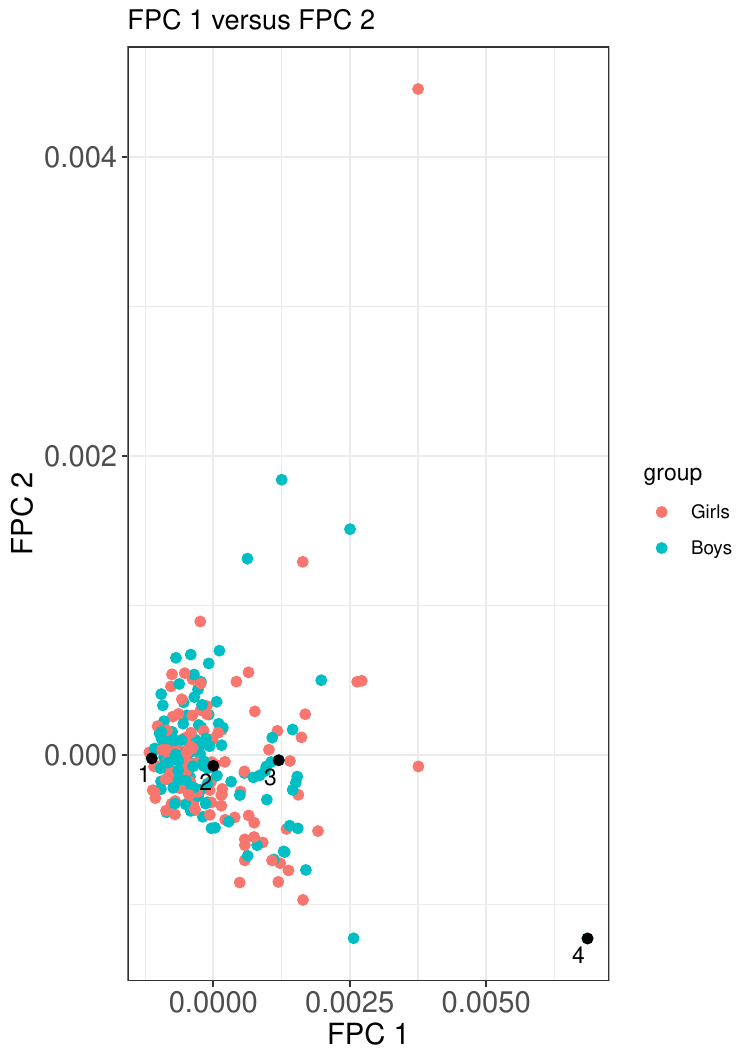}
	\end{subfigure}%
	~
	\begin{subfigure}[t]{0.7\textwidth}
		\centering
		\includegraphics[width=\textwidth]{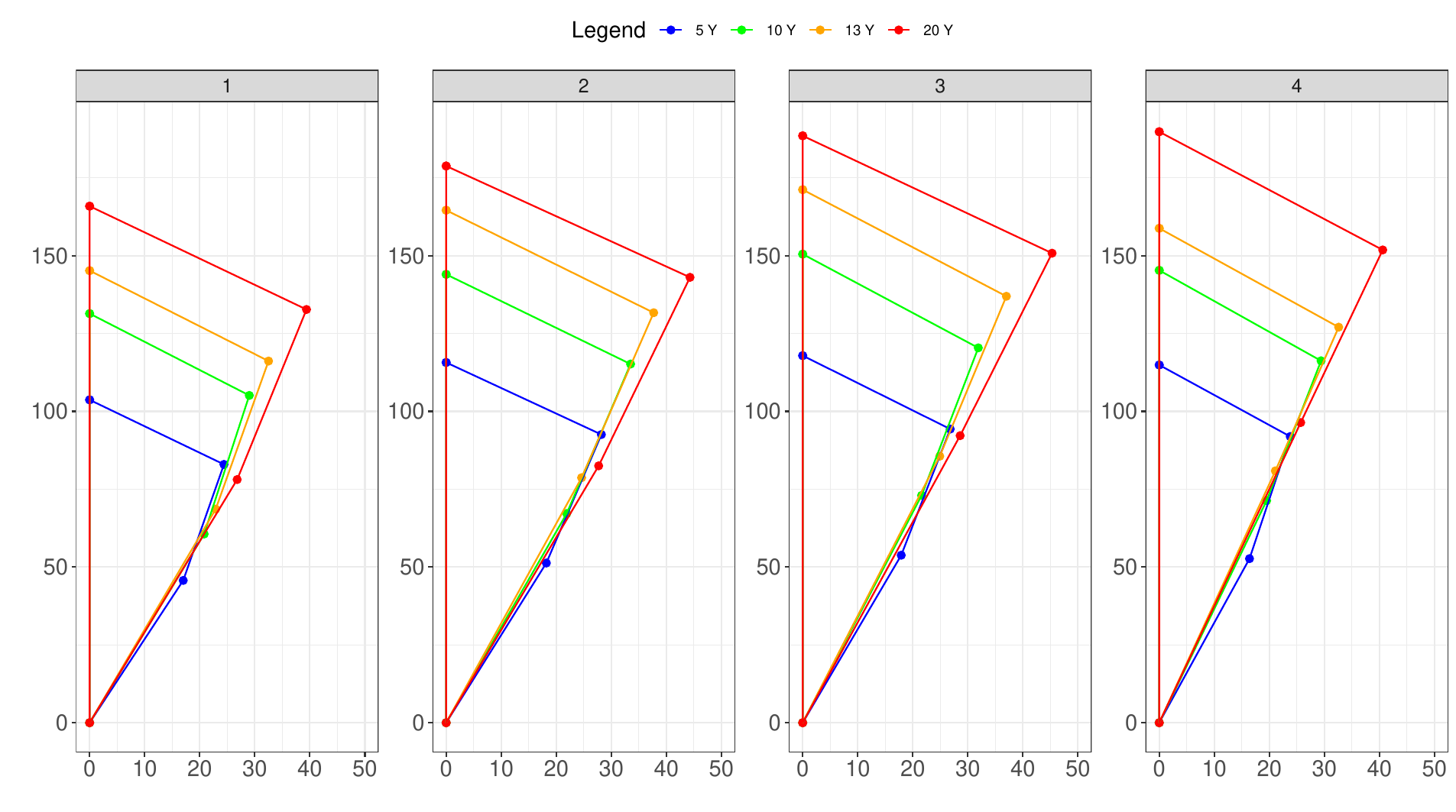}
	\end{subfigure}
	\caption{Left panel: Four selected children (black dots).  Right panel: The shape trajectories $X_i(t), t=5,10,13,20$ for the four  selected children in the left panel. }
	\label{fig: s3}
\end{figure*}

\begin{figure*}[t!]
	\centering
	\hspace{-10mm}
	\begin{subfigure}[t]{0.3\textwidth}
		\centering
		\includegraphics[width=\textwidth]{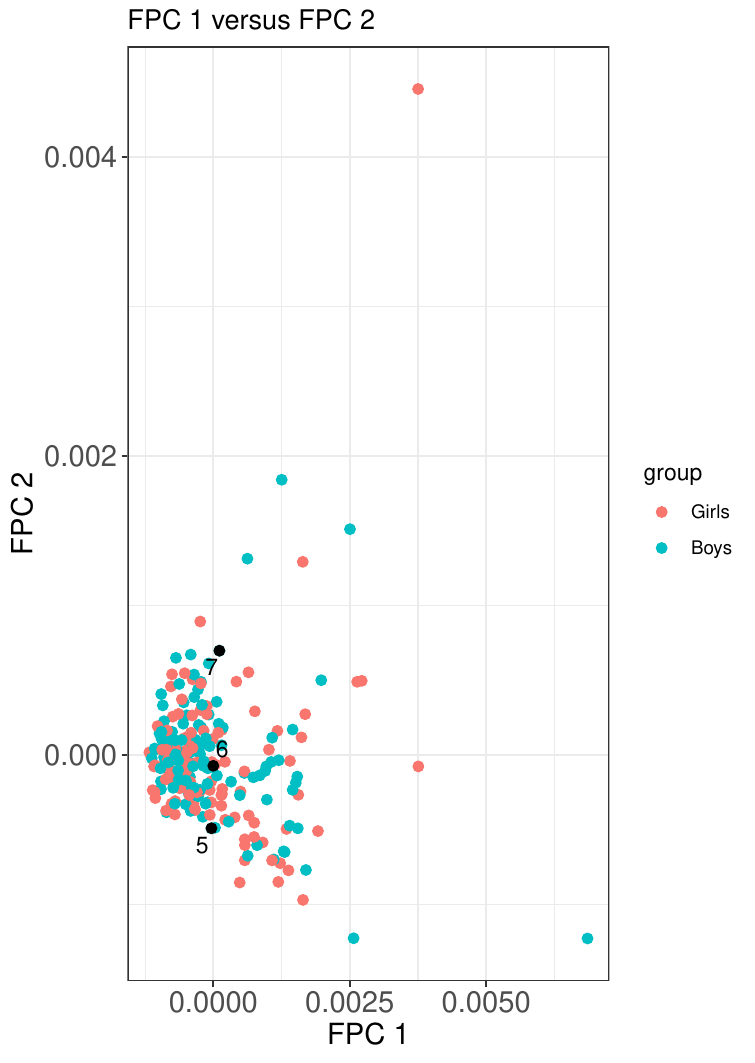}
	\end{subfigure}%
	~
	\begin{subfigure}[t]{0.7\textwidth}
		\centering
		\includegraphics[width=\textwidth]{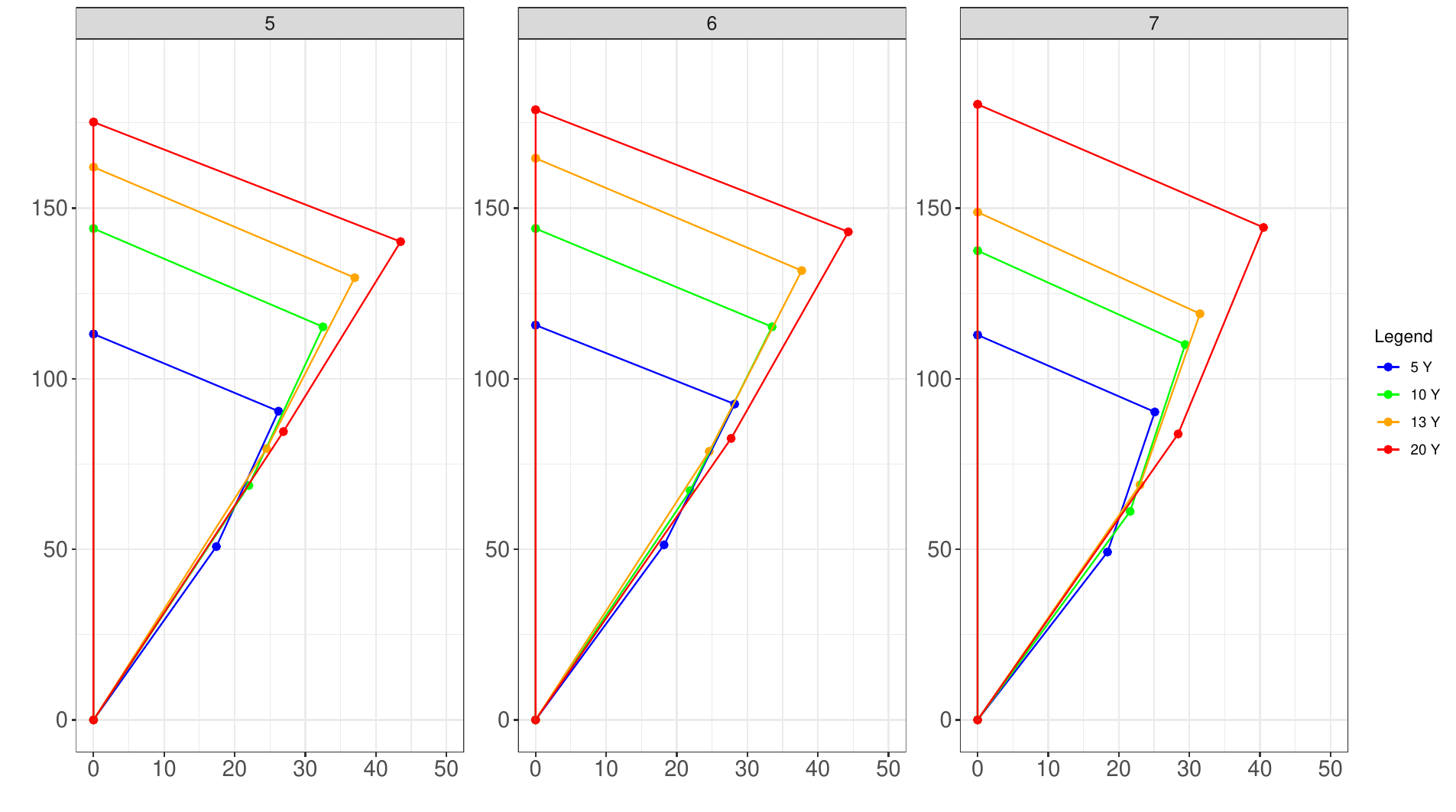}
	\end{subfigure}
	\caption{Left panel: Three selected children (black dots).  Right panel: The shape trajectories $X_i(t), t=5,10,13,20$ for the three selected children in the left panel. }
	\label{fig: s4}
\end{figure*}

\bco
\begin{figure*}[t!]
	\centering
	\hspace{-10mm}
	\begin{subfigure}[t]{0.3\textwidth}
		\centering
		\includegraphics[width=\textwidth]{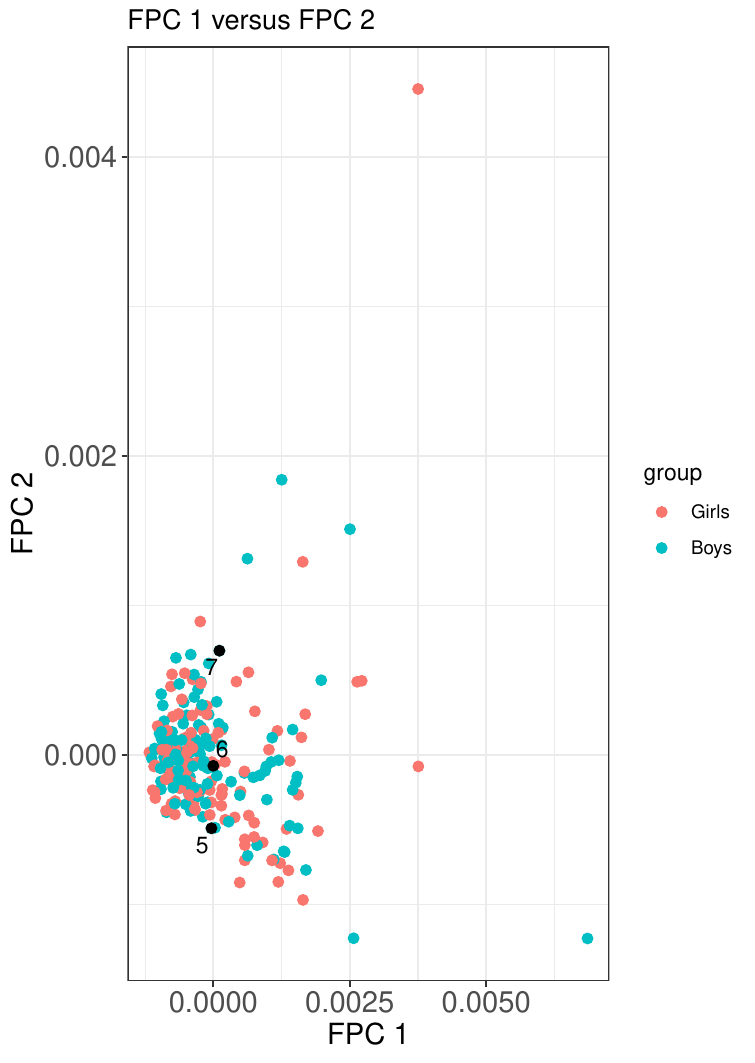}
	\end{subfigure}%
	~
	\begin{subfigure}[t]{0.7\textwidth}
		\centering
		\includegraphics[width=\textwidth]{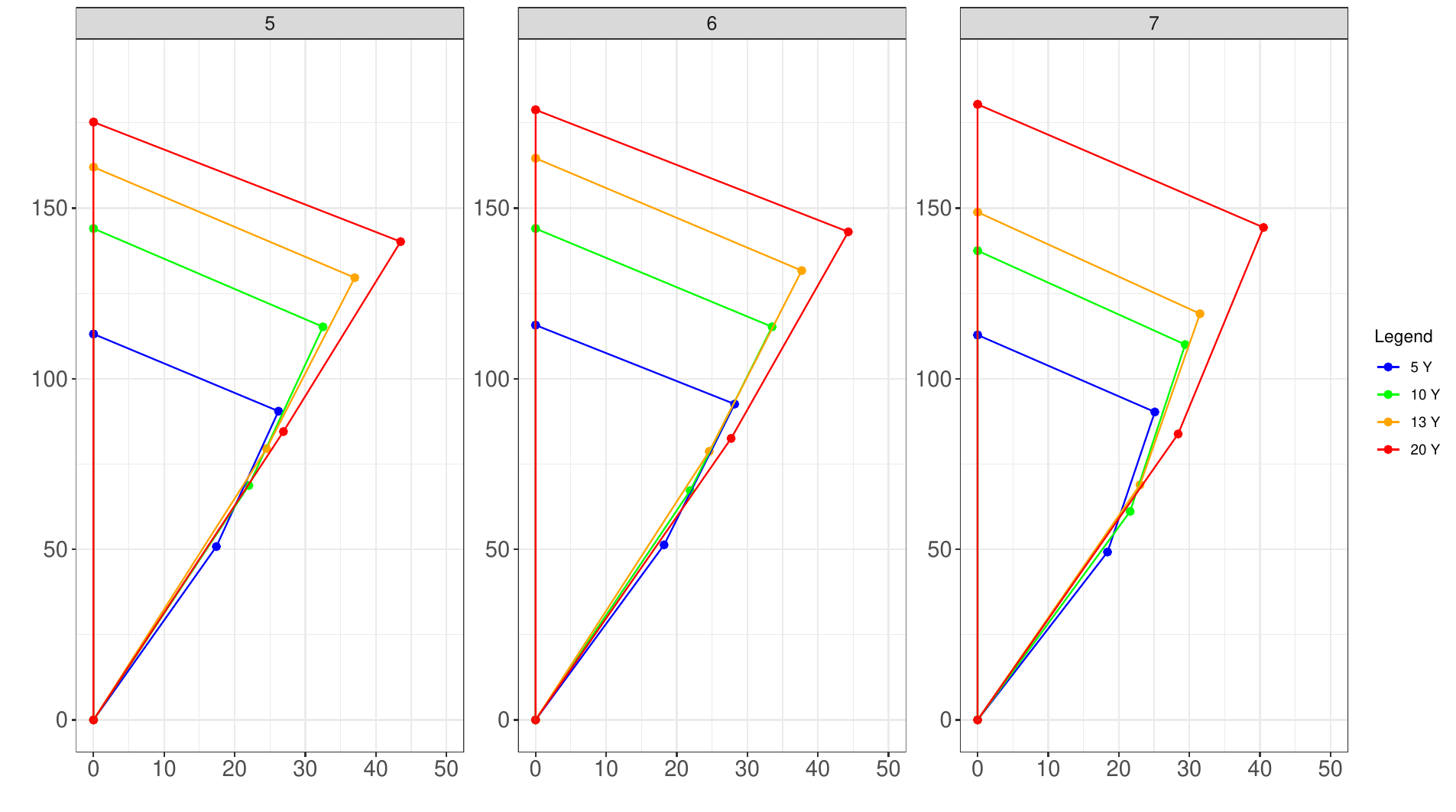}
	\end{subfigure}
	\caption{Left panel: Three selected children. Right panel: The shape trajectories $X_i(t), t=5,10,13,20$ for the three selected children in the left panel. }
	\label{fig: s5}
\end{figure*}
\fi

We also implemented empirical dynamics for these shape data. 
Figure \ref{fig: s6} illustrates the slope function $\hat{\beta}(t)$ and the coefficient of determination $\hat{R}^2(t)$. The latter shows that dynamics in the distance trajectories between the mean shape trajectory and the subject specific shape trajectories can be explained to a large extent by a first order differential equation for the age range 5 to 15 years.  The slope function $\hat{\beta}(t)$ is positive until 17 years of age and negative thereafter, indicating centrifugality of the distance trajectories between infancy and late teenage years, where children's shapes diverge, and centripetality near adulthood.                                                                                                                                     

\begin{figure}[t!]
	\centering
	\includegraphics[width=0.8\textwidth]{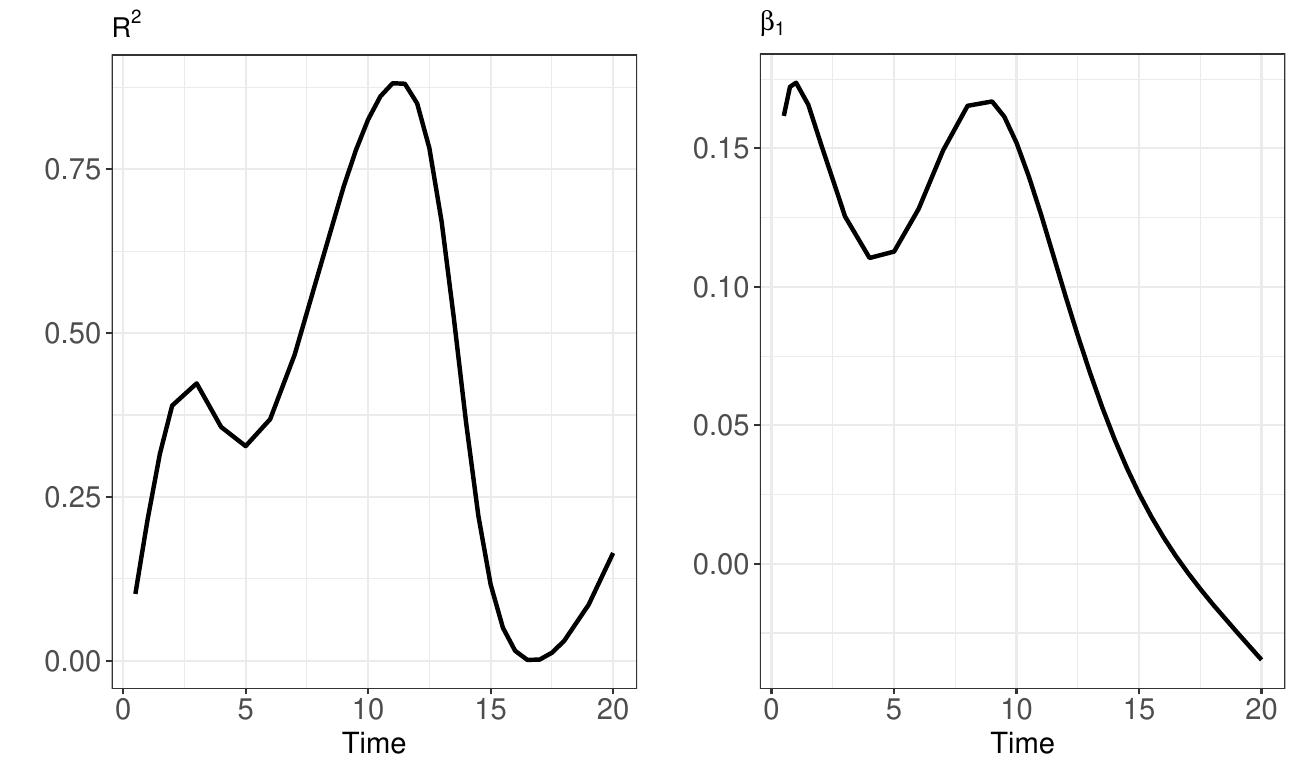}
	\caption{Smooth estimate of the coefficient of determination $R^2(t)$ (left) and the varying coefficient function $\beta(t)$(right),  capturing empirical dynamics of the distance trajectories of the longitudinal growth shape data.}
	\label{fig: s6}
\end{figure}

\section{SIMULATIONS}
\label{sec: simu}

 We illustrate our methods by  simulations using samples of time-varying networks with 20 nodes, inspired by many real world networks that  exhibit community structure, where communities are groups of nodes in a network that show increased within group connectivity and decreased between group connectivity. Existence of community structure is for example  prevalent in traffic networks, particularly the bike networks that we study in our data applications, and also brain networks, social networks and many other areas where networks arise.  We generated the time varying networks as follows.

	{\bf Step 1.} Three groups of time-varying networks with 20 nodes differing  in the community membership of the nodes were generated.  Indexing the nodes of the networks by $1,2,\dots,20$ and the  communities by $C_1,C_2,C_3,C_4$ and $C_5$, the   community membership composition of the nodes for the three groups of networks  was  as follows,
	
	Group 1: Five communities,  $C_1=\{1,2,3,4\}$, $C_2=\{5,6,7,8\}$, $C_3=\{9,10,11,12\}$, $C_4=\{13,14,15,16\}$ and $C_5=\{ 17, 18,19, 20\}$.
	
	Group 2: Four communities,  $C_1=\{1,2,3,4,5,6,7,8\}$, $C_2=\{9,10,11,12\}$, $C_3=\{13,14,15,16\}$ and $C_4=\{ 17, 18,19, 20\}$.
	
	Group 3:  Three communities,  $C_1=\{1,2,3,4\}$, $C_2=\{5,6,7,8\}$ and \\ $C_5=\{9,10,11,12,13,14,15,16,17,18,19,20\}$.
	
	We let the community memberships of the nodes stay fixed in time, while  the edge connectivity strengths $W_{jj'}(t)$ between the communities change with time. The time-varying connectivity weights $W_{jj'}(t)$ between communities $C_j$ and $C_j'$, $j,j' \in \{1,2,3,4,5\}$, that we used when generating the random networks are illustrated in   Figure \ref{fig: sbm1}.   The intra-community connection strengths are higher than the inter-community strengths over the entire time interval. Such dynamic connectivity patterns are  encountered in brain networks \cp{calh:14},    where densely connected brain regions form communities or hubs and inter-hub connectivity often exhibit changing patterns with age.

	{\bf Step 2.}  The network adjacency matrices $A_i(t)$ are generated as:
	
	\begin{equation*}
	(A_i(t))_{k,l} = W_{j,j^{'}}(t) \left \lbrace  \frac{1+\text{sin}(\pi (t+U_{i,kl})V_{i,kl})}{D} \right \rbrace, \ t \in [0,1],
	\end{equation*}
	where $C_j$ is the community membership of node $k$ and $C_j^{'}$ is the community membership of node $l$, $W_{j,j^{'}}(t)$ is the edge connectivity strength between nodes in communities $C_j$ and $C_j'$, $U_{i,kl}$ follows $U(0,1)$ and $V_{i,kl}$ is $1$ if $j=j'$ and sampled uniformly from $\{5,6,\dots,15\}$ if $j \neq j'$. If $j=j'$, we set $D=2$, otherwise $D=4$. Here $U_{i,kl}$ and $V_{i,kl}$ determine  random phase and frequency shifts of the sine function which regulate at what times and how often the edge weights are zero. As $V_{i,kl}$ increases, so does the frequency of the times within $[0,1]$ at which the edge weight is  zero. 	
	The trajectories are represented as graph Laplacians 
	\begin{equation*}
	X_i(t)=D_i(t)-A_i(t), 
	\end{equation*}
	where $D_i(t)$ is a diagonal matrix whose diagonal elements are equal to the sum of the corresponding row elements in $A_i(t)$. Adopting the Frobenius metric in the space of graph Laplacians, the \F \ mean network at time $t$ is  the pointwise average of the graph Laplacians at time $t$, and we obtain the Frobenius distance trajectories of the individual subjects from the \F \ mean trajectory. 

{\bf Step 3.} We carry out FPCA of the distance trajectories generated in Step 2.

The results are shown in Figure \ref{fig: f17}. The proposed method is seen to perform  well in recovering the groups in the scatter plot of the second versus first  functional principal component.  Groups 1 and 2 are found to have closer cluster centers than groups 1 and 3. This is explained by the fact that group 2 is obtained from group 1 by merging $C_1$ and $C_2$ in group 1, which show more similarities than when merging $C_3$, $C_4$ and $C_5$ in group 1 to form group 3.

\begin{figure}
	\centering
	\includegraphics[width=.75\textwidth]{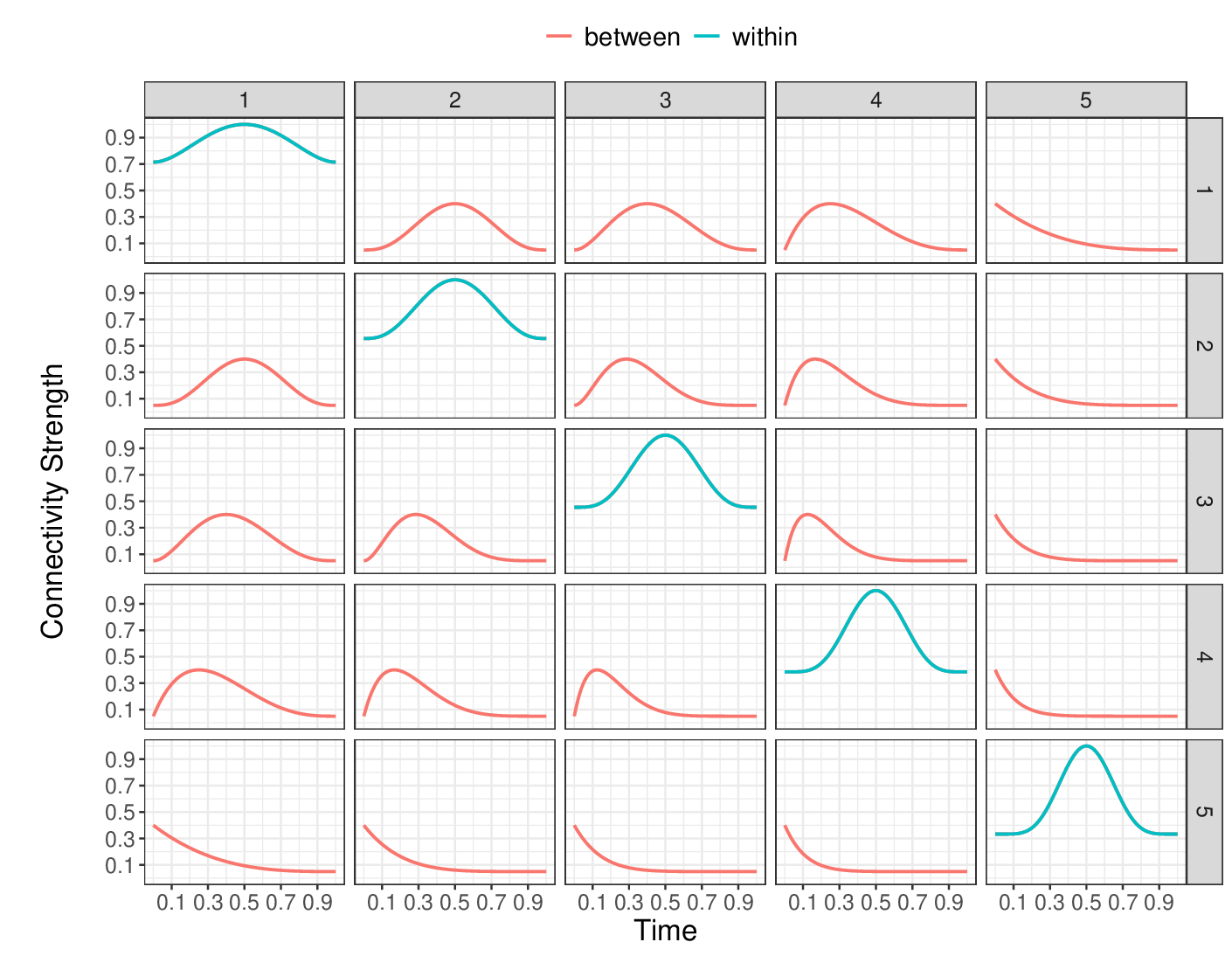}
	\caption{Time varying connectivity weights between the five communities $C_1, C_2, \dots, C_5$.}
	\label{fig: sbm1}
\end{figure}

\begin{figure}[t!]
		\centering
		\includegraphics[width=\textwidth]{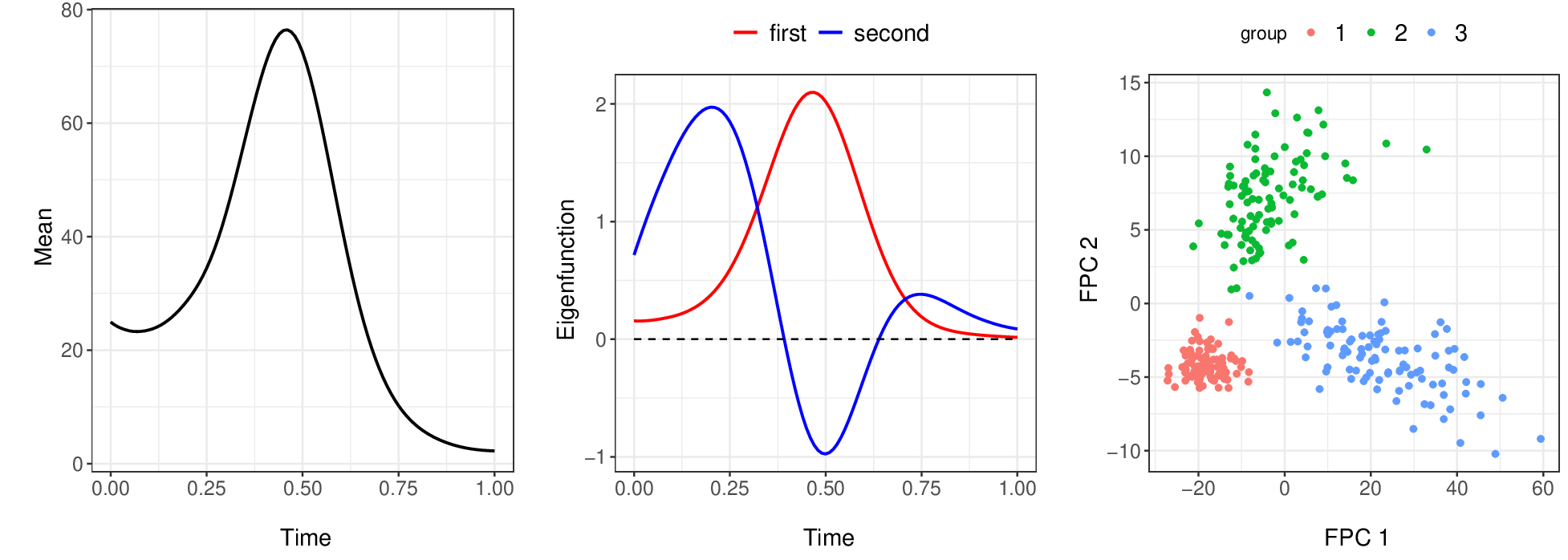}
		\caption{Mean function (left), eigenfunctions (middle) and the scatter plot of second versus first  FPC (right) obtained from functional principal component analysis of the distance trajectories for networks.  The red line corresponds to the first eigenfunction and the blue line to the second.} 
		\label{fig: f17}
	\end{figure}

\section{DISCUSSION}
\label{sec: dis}
We provide a framework for the analysis of time-varying object data,  where the random objects can take values in a general metric space, by defining  a generalized notion of mean function in the object space.  The key to our approach is that  functional data analysis methodology can be applied to squared distance functions of the subject specific curves from the mean function, elucidating the nature of the object time courses, including  empirical dynamics, identifying  clusters, and  detecting extremes and potential outliers in a sample of object trajectories. Another important application is to determine time-specific ranks for subjects, in terms of the distance of the subject's trajectory from the mean trajectory; such ranks have important applications in health monitoring, for example in neurodevelopment \citep{dosm:12}. 
The FPCs that we obtain for time-varying random objects can also be used for regression analysis, where object time courses are responses or possibly predictors. 

{As pointed out by a referee, another possible interpretation of  empirical dynamics can be gained with the notion of antimeans. For compact metric spaces, when the extreme values of the \F \ function $E(d^2(Y,\o))$ are attained in $\Omega$, the set of maximizers of the \F \ function forms a newly introduced notion of location parameter called the \F \ \emph{antimean} set \citep{wang:20,patr:16,patr:16:2}. Just like  \F \ means, \F \ antimeans can form useful statistics for describing and comparing samples of object data, as has been illustrated for  samples of shapes of lily flowers  \citep{patr:16:2}. For data that reside in compact manifolds where the notion of the data center is obscure, \F \ antimeans could provide constructive data summaries. Hence once could frame 
 empirical dynamics with respect to not only the \F \ means, but also the \F \ antimeans, which might uncover insightful findings in terms of regression to the mean or the antimean, leading to multiple points of interest. In such situations, in the context of centrifugality, one might tend to move closer to the \F \ antimean, for which \emph{anticentripetality} might be a fitting term.}

 Our data examples include samples of time-varying univariate probability distributions and samples of time evolving networks. In practice the object functions may  not be fully observed, but instead are  observed on a more or less dense time grid, possibly with noise.  In such situations one can opt for  smoothing the object trajectories if the observation grid is sufficiently dense. To implement preliminary smoothing and interpolation,  
 \F \ regression provides a possible  option  \cp{mull:19:3}. While the role of smoothing individual trajectories in FDA is well understood in the Euclidean case  \cp{hall:07:3, zhan:07}, it  remains an open problem to investigate its properties in the much more general setting of longitudinal object data.  
 An even bigger challenge that is also left  for future research is the case where measurement grids are sparse and irregular, {a problem that was recently studied in \cite{lin:20,dai:20} for data on Riemannian manifolds.}

\single
\references

\double
\section*{ONLINE SUPPLEMENT}

\subsection*{A.1 Overview of the Theoretical Derivations}
\label{supp1}
For an element $\o$ taking values in $\O$ and $s \in [0,1]$, define 
\be 
U_n(\o,s) = \frac{1}{n} \sum_{i=1}^n  f_{\o,s}(X_i) \ \text{and} \ U(\o,s) =  E\lbrace f_{\o,s}(X)  \rbrace \la{U}
\ee
where $f_{\o,s}(x)=d^2(x(s),\o)-d^2(x(s),\mu(s))$. In our setting, $\hat{\mu}(s)$ is the minimizer of $U_n(\o,s)$ and $\mu(s)$ is the minimizer of $U(\o,s)$. 
A critical step in establishing uniform convergence of the plug-in estimator of the \F \ covariance surface $\tilde{C}$ in (6)  given by $\hat{C}$ (7) is to find an upper bound for the quantity  $E\left(\operatornamewithlimits{sup}_{s \in [0,1], \om d(\o,\mu(s)) < \delta}  \left|U_n(\o,s)-U(\o,s)\right|\right)$ for $\delta > 0$. 
For this, it is convenient to consider function classes 
\begin{equation}
\label{f_delta}
\mathcal{F}_{\delta}=\lbrace f_{\o,s} :  s \in [0,1], \, d(\o,\mu(s)) < \delta \rbrace
\end{equation} and to apply empirical process theory.
An envelope function for this class is the constant function $F(x)=2M\delta$, where $M$ is the diameter of $\O$. The $L_{2}$ norm of the envelope function is $||F||_{2}=2M\delta$. 

\bco

Throughout this Supplement we use the abbreviation 
$\operatornamewithlimits{sup}_{d(\o,\mu(s)) < \delta} H(\o,s)$ for 
$\operatornamewithlimits{sup}_{\o \in \O: \, d(\o,\mu(s)) < \delta} H(\o,s)$ for 
any function $H$ and analogously for $\operatornamewithlimits{inf}$ and 
$\operatornamewithlimits{sup}_{d(\o,\mu(s)) >\epsilon} H(\o,s)$, 
in order to unclutter the notation. 

\fi

By Theorem 2.14.2 of \cite{well:96},
\begin{equation}
\label{eq: tail_bound2}
E\left( \operatornamewithlimits{sup}_{s \in [0,1],\om d(\o,\mu(s)) < \delta} \left|U_n(\o,s)-U(\o,s)\right|\right) \leq \frac{2M\delta \ J_{[]}(1,\mathcal{F}_{\delta},L_{2}(P))}{\sqrt{n}},
\end{equation}
where $J_{[]}(1,\mathcal{F}_{\delta},L_{2}(P))$ is the bracketing integral of the function class $\mathcal{F}_\delta$, which is 
\begin{equation*}
J_{[]}(1,\mathcal{F}_{\delta},L_{2}(P))= \int_{0}^{1} \sqrt{1+\log N(\eps ||F||_{2},\mathcal{F}_{\delta},L_2(P))} d\eps.
\end{equation*}
Here $N(\eps ||F||_{2},\mathcal{F}_{\delta},L_2(P))$ is the minimum number of balls of radius $\eps ||F||_{2}$ required to cover the function class $\mathcal{F}_{\delta}$ under the $L_2(P)$ norm.  
Assumptions (A2)-(A4)  imply 
\begin{equation}
\label{eq: entropy}
J_{[]}(1,\mathcal{F}_{\delta},L_{2}(P))=O(\sqrt{\log{1/\delta}}) \quad \text{as} \quad \delta \rightarrow 0;
\end{equation}
see Lemma \ref{lma:entropy} in Section \hyperref[supp3]{A.3}.

Equation \eqref{eq: entropy} provides 
a key result for obtaining the rate of convergence of $\hat{\mu}$. One can show that under assumptions (A1)-(A4),
\begin{equation}
\label{eq: rate}
\sup_{s \in [0,1]} d(\hat{\mu}(s),\mu(s))=O_P\left(\left(\frac{\sqrt{\log{n}}}{n}\right)^{\frac{1}{2(\beta-1)}}\right), 
\end{equation}
with $\beta$ as in (A2); see Lemma  \ref{lma: rate} in Section \hyperref[supp3]{A.3}.

In Lemma \ref{lma: unif_var} in Section \hyperref[supp2]{A.2} we establish the uniform convergence of the process $\frac{1}{n} \sum_{i=1}^{n} d^2(X_i(s),\hat{\mu}(s))$ to $\frac{1}{n} \sum_{i=1}^{n} d^2(X_i(s),\mu(s))$. This  is the main step to derive the  uniform convergence of the sample based \F \ covariance surface $\hat{C}(s,t)$ in (7)  to the oracle \F \ covariance surface $\tilde{C}(s,t)$ in (6)  (see  Lemma \ref{lma: unif_cov}).

\subsection*{A.2 Proofs of Theorem 1 and 2}
\label{supp2}
\begin{proof}[Proof of Theorem 1]
	For Theorem 1, we need two additional Lemmas:
	\begin{Lemma} [Convergence of sample \F \ variance function]  
		\label{lma: unif_var}
		Under assumptions (A1)-(A4),
		\begin{equation*}
		\sup_{s \in [0,1]} \sqrt{n} \left|\frac{1}{n} \sum_{i=1}^{n} \lbrace V_i(s)-V_i^*(s)\rbrace\right| = o_P(1).
		\end{equation*}
	\end{Lemma}
	\begin{proof}[Proof of Lemma \ref{lma: unif_var}]
		Let $\eps > 0$ and $\delta > 0$. Observe 
		\begin{align}
		& P\left(\sup_{s \in [0,1]} \sqrt{n} \left|\frac{1}{n} \sum_{i=1}^{n} d^2(X_i(s),\hat{\mu}(s))-\frac{1}{n} \sum_{i=1}^{n} d^2(X_i(s),\mu(s))\right| > \eps \right) \nonumber \\  & = P\left(\inf_{s \in [0,1]} \left(\frac{1}{n} \sum_{i=1}^{n} d^2(X_i(s),\hat{\mu}(s))-\frac{1}{n} \sum_{i=1}^{n} d^2(X_i(s),\mu(s))\right) <  -\frac{\eps}{\sqrt{n}} \right) \nonumber \\ & \leq  P\left(\inf_{s \in [0,1], \om d(\o,\mu(s)) < \delta} \left(\frac{1}{n} \sum_{i=1}^{n} d^2(X_i(s),\o)-\frac{1}{n} \sum_{i=1}^{n} d^2(X_i(s),\mu(s))\right) <  -\frac{\eps}{\sqrt{n}} \right) \nonumber \\ & + P \left(\sup_{s \in [0,1]} d(\hat{\mu}(s),\mu(s)) > \delta\right) \nonumber \\ & \leq P\left(\sup_{s \in [0,1], \om d(\o,\mu(s)) < \delta} \left|U_n(\o,s)-U(\o,s)\right| > \frac{\eps}{\sqrt{n}}\right) + P \left(\sup_{s \in [0,1]} d(\hat{\mu}(s),\mu(s)) > \delta\right) \label{eq: inequality}.
		\end{align}
		This is because $\inf_{s \in [0,1], \om d(\o,\mu(s)) < \delta} \lbrace E(d^2(X(s),\o))-E(d^2(X(s),\mu(s)) \rbrace =0$ implies 
		{\begin{align*}
			& \sup_{s \in [0,1], \om d(\o,\mu(s)) < \delta} \left|U_n(\o,s)-U(\o,s)\right| \\ & \geq \left| \inf_{s \in [0,1],\om d(\o,\mu(s)) < \delta} U_n(\o,s)-\inf_{s \in [0,1], \om  d(\o,\mu(s)) < \delta} U(\o,s)\right|.
			\end{align*}}
		Continuing from equation \eqref{eq: inequality}, we find that the second term goes to zero due to Lemma \ref{lma: mean_uniform} in Section \hyperref[supp3]{A.3}. Using Markov inequality, equation \eqref{eq: tail_bound2} and  Lemma \ref{lma:entropy} in Section \hyperref[supp3]{A.3} , by choosing $\delta$ sufficiently small, the first term is upper bounded by 
		\begin{equation}
		\frac{\sqrt{n}E\left(\sup_{s \in [0,1], \om d(\o,\mu(s)) < \delta} \left|U_n(\o,s)-U(\o,s)\right|\right)}{\epsilon} \leq \frac{2M\delta J_{[]}(1,\mathcal{F}_{\delta},L_{2}(P))}{\epsilon}.
		\end{equation}
		Given any $\eps > 0, h > 0$, we aim to show that there exists a sufficiently large integer $N$ such that for all $n \geq N$ we have 
		$$P\left(\sup_{s \in [0,1]} \sqrt{n} \left|\frac{1}{n} \sum_{i=1}^{n} d^2(X_i(s),\hat{\mu}(s))-\frac{1}{n} \sum_{i=1}^{n} d^2(X_i(s),\mu(s))\right| > \eps \right) < h.$$ To do this we first choose $\delta$ small enough such that $\frac{2M\delta J_{[]}(1,\mathcal{F}_{\delta},L_{2}(P))}{\epsilon} < \frac{h}{2}$ and then  choose $N$ large enough such that $P \left( \sup_{s \in [0,1]} d(\hat{\mu}(s),\mu(s)) > \delta\right) < \frac{h}{2}$. This completes the proof.
	\end{proof}
	{\begin{Lemma}[Convergence of sample \F \ covariance surface]  
		\label{lma: unif_cov}
		Under assumptions (A1)-(A4),
		\begin{equation*}
		 \sup_{s,t \in [0,1]}\left|\hat{C}(s,t)-\tilde{C}(s,t)\right|=O_P\left(\max \left\lbrace\left(\frac{\sqrt{\log{n}}}{n}\right)^{\frac{1}{2(\beta-1)}},\frac{1}{\sqrt{n}}\right\rbrace \right).
		\end{equation*}
	\end{Lemma}}
	{\begin{proof}[Proof of Lemma \ref{lma: unif_cov}] Observe that
		\begin{equation*}
		 \sup_{s,t \in [0,1]}\left|\hat{C}(s,t)-\tilde{C}(s,t)\right|  \leq \sup_{s,t \in [0,1]} A_1(s,t)+ 2M^2  \sup_{s \in [0,1]} A_2(s),
		\end{equation*}
		where $A_1(s,t)$ and $A_2(s)$ are given by
		\begin{equation*}
		A_1(s,t)=\left|\frac{1}{n} \sum_{i=1}^{n} d^2(X_i(s),\hat{\mu}(s)) d^2(X_i(t),\hat{\mu}(t))-\frac{1}{n} \sum_{i=1}^{n} d^2(X_i(s),{\mu}(s)) d^2(X_i(t),{\mu}(t))\right|,
		\end{equation*}
		\begin{equation*}
		A_2(s)=\left|\frac{1}{n} \sum_{i=1}^{n} d^2(X_i(s),\hat{\mu}(s)) -\frac{1}{n} \sum_{i=1}^{n} d^2(X_i(s),{\mu}(s)) \right|.
		\end{equation*}
        By using the triangle inequality
        \begin{align*}
            & A_1(s,t) \\ \leq & \left|\frac{1}{n} \sum_{i=1}^{n} d^2(X_i(s),\hat{\mu}(s)) d^2(X_i(t),\hat{\mu}(t))-\frac{1}{n} \sum_{i=1}^{n} d^2(X_i(s),{\mu}(s)) d^2(X_i(t),{\hat{\mu}}(t))\right| \\ & + \left|\frac{1}{n} \sum_{i=1}^{n} d^2(X_i(s),{\mu}(s)) d^2(X_i(t),\hat{\mu}(t))-\frac{1}{n} \sum_{i=1}^{n} d^2(X_i(s),{\mu}(s)) d^2(X_i(t),{{\mu}}(t))\right| \\ \leq & 4 M^3  \sup_{s \in [0,1]} d(\hat{\mu}(s),\mu(s)).
        \end{align*}
		Hence by Lemma \ref{lma: rate}, $\sup_{s,t \in [0,1]} A_1(s,t)= O_P\left(\left(\frac{\sqrt{\log{n}}}{n}\right)^{\frac{1}{2(\beta-1)}}\right)$. Using Lemma \ref{lma: unif_var} we have $\sqrt{n} \sup_{s \in [0,1]} A_2(s)=o_P(1)$ which completes our proof.
	\end{proof}}
	 \noindent For Theorem 1, observe that 
 \begin{equation}
 \label{eq: cov_break}
      \sup_{s,t \in [0,1]}\left|\hat{C}(s,t)-{C}(s,t)\right| \leq  \sup_{s,t \in [0,1]}\left|\hat{C}(s,t)-\tilde{C}(s,t)\right|+ \sup_{s,t \in [0,1]}\left|\tilde{C}(s,t)-{C}(s,t)\right|.
 \end{equation}
  {To upper bound the second term, observe that ${C}(s,t)$ is the covariance function of a second order stochastic process $\{Y(s)\}_{s \in [0,1]}$ where $Y(s)=d^2(X(s),\mu(s))$ for $s \in [0,1]$  and has mean function $E(Y(s))=E(d^2(X(s),\mu(s)))$. The process $\{Y(s)\}_{s \in [0,1]}$ can also be viewed  as a random element $\mathcal{Y} \in \mathcal{L}^2[0,1]$ with mean  $m=E\mathcal{Y}$ and  covariance operator generated by the covariance function $C(s,t)$ denoted by $\mathcal{C}$. Let $\mathcal{Y}_1,\mathcal{Y}_2, \dots,\mathcal{Y}_n$ be i.i.d.  realizations of $\mathcal{Y}$ from which one obtains  the sample covariance operator $\tilde{\mathcal{C}}$, the covariance operator generated by $\tilde{C}(s,t)$. By the boundedness of the metric, we have $E(||\mathcal{Y}||^4) < \infty$ and therefore by Theorem 8.1.2 of \cite{hsin:15}, $\sqrt{n}(\tilde{\mathcal{C}}-\mathcal{C})$ converges weakly to a Gaussian random element with mean zero and covariance operator given by
		\begin{equation}
		\mathcal{R}= E\left(((\mathcal{Y} -m)\otimes (\mathcal{Y}-m)-\mathcal{C})\otimes_{HS}((\mathcal{Y} -m)\otimes (\mathcal{Y}-m)-\mathcal{C})\right).
		\end{equation}
		By the continuous mapping theorem $\sqrt{n} \sup_{s,t \in [0,1]} \lvert \tilde{C}(s,t)-C(s,t)\rvert = O_P(1)$. In conjunction with Lemma \ref{lma: unif_cov} which upper bounds the first term in equation \eqref{eq: cov_break} we have
  \begin{equation*}
		 \sup_{s,t \in [0,1]}\left|\hat{C}(s,t)-{C}(s,t)\right|=O_P\left(\max \left\lbrace\left(\frac{\sqrt{\log{n}}}{n}\right)^{\frac{1}{2(\beta-1)}},\frac{1}{\sqrt{n}}\right\rbrace \right).
		\end{equation*}
  }

		\bco
		{\it Step 2: Convergence of the Marginals to a Multivariate Normal Distribution.}
		For any integer $k$ and any $(s_1,t_1), (s_2,t_2), \hdots, (s_k,t_k) \in [0,1]^2$, let $C_i=C(s_i,t_i),$ $\hat{C}_i=\hat{C}(s_i,t_i)$ and $\tilde{C}_i=\tilde{C}(s_i,t_i)$, and further  $C^k=\left(C_1,C_2,\hdots,C_k\right)$, $\hat{C}^k=\left(\hat{C}_1,\hat{C}_2,\hdots,\hat{C}_k\right)$ and $\tilde{C}^k=\left(\tilde{C}_1,\tilde{C}_2,\hdots,\tilde{C}_k\right)$. Observe that
		\begin{equation*}
		\sqrt{n} \left(\hat{C}^k-C^k\right)=\sqrt{n} \left(\hat{C}^k-\tilde{C}^k\right)+\sqrt{n} \left(\tilde{C}^k-C^k\right),
		\end{equation*}
		where the first term converges to zero in probability as a consequence of Lemma \ref{lma: unif_cov} and the second term converges in distribution to a multivariate normal distribution with mean zero and covariance matrix $\Sigma^k=(\Sigma^k_{ij})$,  where $\Sigma^k_{ij}=\mathcal{R}_{(s_i,t_i),(s_j,t_j)}$ as a consequence of Step 1. Therefore by Slutsky's theorem for any integer $k$ and any $(s_1,t_1), (s_2,t_2), \hdots, (s_k,t_k) \in [0,1]^2$, 	$\sqrt{n} \left(\hat{C}^k-C^k\right)$ converges in distribution to $N(0,\Sigma^k)$.\\
		\fi
		\bco
		{\it Step 3: Uniform Convergence of the Process $ \left(\hat{C}(s,t)-C(s,t)\right)$.}
		\\Let $(s_1,t_1),(s_2,t_2) \in [0,1]^2$ be such that $|s_1-s_2|+|t_1-t_2| < \delta$ for some $\delta > 0$. We need to show that for any $S>0$ and as $\delta \rightarrow 0$,
		\begin{equation}
		\limsup_{n \rightarrow \infty} P\left(\sup_{|s_1-s_2|+|t_1-t_2| < \delta}\sqrt{n} \left|\hat{C}(s_1,t_1)-C(s_1,t_1) -\hat{C}(s_2,t_2)+C(s_2,t_2)\right| > 2S\right) \rightarrow 0.
		\end{equation}
		This is true because
		\begin{align*}
		& \limsup_{n \rightarrow \infty} P\left(\sup_{|s_1-s_2|+|t_1-t_2| < \delta}\sqrt{n} \left|\hat{C}(s_1,t_1)-C(s_1,t_1) -\hat{C}(s_2,t_2)+C(s_2,t_2)\right| > 2S\right) \\ & \leq A_\delta+ B_\delta,
		\end{align*}
		where $A_\delta=\limsup_{n \rightarrow \infty} P\left(2\sup_{(s,t) \in [0,1]}\sqrt{n} \left|\hat{C}(s,t)-\tilde{C}(s,t)\right| > S\right)$  and $B_\delta=$ \newline $\limsup_{n \rightarrow \infty} P\left(\sup_{|s_1-s_2|+|t_1-t_2| < \delta}\sqrt{n} \left|\tilde{C}(s_1,t_1)-C(s_1,t_1) -\tilde{C}(s_2,t_2)+C(s_2,t_2)\right| > S\right)$. \newline By Lemma \ref{lma: unif_cov}, we have $A_{\delta}=0$ for any $S>0$ and from Step 1, the uniform asymptotic equicontinuity of the process $\sqrt{n} \left(\tilde{C}(s,t)-C(s,t)\right)$ implies  $B_\delta \rightarrow 0$ as $\delta \rightarrow 0$ for any $S>0$.\\
	
	From Steps 1-3 and Theorems 1.5.4 and 1.5.7 in \cite{well:96} the result follows. 
\fi
\end{proof}

\begin{proof}[Proof of Theorem 2]
	For any $j$, perturbation theory as stated for example in  Lemma 4.3 of  \cite{bosq:00} gives $|\hat{\lambda}_j-\lambda_j| \leq \sup_{s,t \in [0,1]}\left|\hat{C}(s,t)-C(s,t)\right|$. Along with Theorem 1 this implies 
	{$|\hat{\lambda}_j-\lambda_j|=O_P\left(\max \left\lbrace\left(\frac{\sqrt{\log{n}}}{n}\right)^{\frac{1}{2(\beta-1)}},\frac{1}{\sqrt{n}}\right\rbrace \right).$}
	For any $j$, one also has $$\sup_{s \in [0,1]} \left| \hat{\phi}_j(s)-\phi_j(s)\right| \leq 2\sqrt{2}\delta_j^{-1}\sup_{s,t \in [0,1]}\left|\hat{C}(s,t)-C(s,t)\right|. $$  Applying Theorem 1 and assumption (A5) leads to  
	{\begin{equation} \label{eq: eig}
	\sup_{s \in [0,1]} \left| \hat{\phi}_j(s)-\phi_j(s)\right| =O_P\left(\frac{1}{\delta_j}\max \left\lbrace\left(\frac{\sqrt{\log{n}}}{n}\right)^{\frac{1}{2(\beta-1)}},\frac{1}{\sqrt{n}}\right\rbrace \right).
	\end{equation}}
	
	Next we show the convergence of the scores $\hat{B}_{ij}$ to their population targets $B_{ij}$.  For this consider the functions $g_s(x)=d^2(x(s),\mu(s))$ and the function class $\mathcal{F}=\{g_s : s \in [0,1]\}$. Then 
	\begin{align*}
	|f_s(x)-f_t(x)|  \leq  2M (d(x(s),x(t)) +d(\mu(s),\mu(t))).
	\end{align*}
	Following the arguments in the proof of Lemma  \ref{lma: rate} in Section \hyperref[supp3]{A.3}, under assumptions (A2) and (A3), whenever $|s-t|<\left(\frac{\rho}{K}\right)^{\beta/\alpha}$, we have with $K= \left \lbrace 4MD^{-1}E(G(X)) \right \rbrace^{1/\beta}$
	\begin{equation*}
	d(\mu(s),\mu(t)) \leq K |s-t|^{\frac{\alpha}{\beta}},
	\end{equation*}
	where $\alpha$, $\beta$, $\rho$, $D$ and $G(\cdot)$ are as defined in assumptions (A2) and (A3). For $0 < u < 1$,  if $s_1,s_2, \hdots, s_K$ is a  $\left(\frac{\rho u}{K}\right)^{\frac{\beta}{\alpha}}$-net for $[0,1]$ {with metric $|\cdot|$} , then for any $s \in [0,1]$ one can find $s_j \in [0,1]$ such that
	\begin{equation*}
    |f_s(x)-f_{s_j}(x)| \leq 2M \left \lbrace G(x)  \left(\frac{\rho u}{K}\right)^{\beta} +K  \left(\frac{\rho u}{K}\right) \right \rbrace  \leq G'(x) u,
    \end{equation*}
	where $G'(x)=2M \left \lbrace G(x)  \left(\frac{\rho }{K}\right)^{\beta} +K  \left(\frac{\rho }{K}\right) \right \rbrace$. Hence the brackets $[f_{s_i}\pm G'u]$ cover the function class $\mathcal{F}$ and are of size $2 ||G'||_{L_2}u$, implying that for some constant $C'$,
	\begin{equation*}
	N_{[]}(u,\mathcal{F},L_2(P)) \leq N\left(\left(\frac{\rho u}{2||G'||_{L_2}K}\right)^{\frac{\beta}{\alpha}},[0,1],|\cdot|\right) \leq \frac{C'}{u^{\frac{\beta}{\alpha}}}.
	\end{equation*}
	Here  $N_{[]}(\eps,\mathcal{F},L_2(P))$ is the bracketing number,  defined in \cite{well:96} as  the minimum number of $\eps$-brackets needed to cover $\mathcal{F}$,  where an $\eps$-bracket is composed of pairs of functions $[l,u]$ such that $||l-u||_{L_2} < \eps$.  Then 
	\begin{align*}
	& \int_{0}^{1} \sqrt{\log N_{[]}(u,\mathcal{F},L_2(P)) } du \\ \leq & \int_{0}^{1} \sqrt{\log C' - \frac{\beta}{\alpha}\log u } \ du \leq   \sqrt{\log C'} + \int_{0}^{1} \sqrt{-\frac{\beta}{\alpha}\log u}\ du \\ = &  \sqrt{\log C'} + \sqrt{\frac{\beta}{\alpha}}\int_{0}^{\infty} x^{1/2}e^{-x}dx\,\,  < \,\, \infty.
	\end{align*}
	By Theorem 3.1 in \cite{ossi:87} we have that $\mathcal{F}$ has the Donsker property and therefore $\frac{1}{\sqrt{n}} \sum_{i=1}^{n} (V^\ast_i(t)-\nu^\ast(t))$ converges weakly to a Gaussian process limit with zero mean and covariance function 
	$H(s,t)=\mathrm{Cov}(V^\ast(s),V^\ast(t)).$
	By using Lemma \ref{lma: unif_var} and the Slutsky's theorem, $\frac{1}{\sqrt{n}} \sum_{i=1}^{n} (V_i(t)-\nu^\ast(t))$ is seen to converge weakly to a Gaussian process limit with zero mean and covariance function $H(s,t)$, and by 
	continuous mapping
	\begin{equation} \label{eq: var_boud}\sup_{t \in [0,1]}\left|\frac{1}{n} \sum_{i=1}^{n} (V_i(t)-\nu^\ast(t))\right|=O_P(n^{-1/2}).\end{equation}
	Observing 
	\begin{align*}
	& |\hat{B}_{ij}-B_{ij}| \\ = & \Big | \int_{0}^{1} \left(V_i(t)-\frac{1}{n}\sum_{k=1}^{n}V_k(t) \right)\hat{\phi}_j(t) dt-\int_{0}^{1} (V^\ast_i(t)-\nu^\ast(t))\phi_j(t) dt \Big | \\ \leq & \Big | \int_{0}^{1} \left(V_i(t)-\frac{1}{n}\sum_{k=1}^{n}V_k(t) \right) \left(\hat{\phi}_j(t)-\phi_j(t)\right) dt \Big | \\ & \quad + \Big  |\int_{0}^{1} \left(V_i(t)-\frac{1}{n}\sum_{k=1}^{n}V_k(t)- V^\ast_i(t)+\nu^\ast(t)\right) \phi_j(t) dt\Big | \\ \leq & 2 M^2 \sup_{t \in [0,1]} \left| \hat{\phi}_j(t)-\phi_j(t)\right| + \sup_{t \in [0,1]} \left| V_i(t)-V^\ast_i(t) \right|+ \sup_{t \in [0,1]} \left| \frac{1}{n}\sum_{i=1}^{n}V_i(t)-\nu^\ast(t) \right|
	\end{align*}
	and noting that $\sup_{t \in [0,1]} |V_i(t)-V^\ast_i(t)| \leq 2M \sup_{t \in [0,1]} d(\hat{\mu}(t),\mu(t))$, 
	Lemma \ref{lma: rate} in\hyperref[supp3]{A.3}  below and  \eqref{eq: var_boud} and \eqref{eq: eig} imply 
	{\begin{equation*}
	|\hat{B}_{ij}-B_{ij}|=O_P\left(\max\left\lbrace \frac{1}{\delta_j},1\right\rbrace \max \left\lbrace\left(\frac{\sqrt{\log{n}}}{n}\right)^{\frac{1}{2(\beta-1)}},\frac{1}{\sqrt{n}}\right\rbrace \right).
	\end{equation*}}
	
\end{proof}
\subsection*{A.3 Auxiliary Lemmas} 
\label{supp3}

Lemma \ref{lma: mean_uniform} is established  in \ci{mull:19:1}  and is reproduced here for convenience. Lemmas \ref{lma:entropy} and \ref{lma: rate} are established here under weaker conditions by modifying the proofs in \ci{mull:19:1}, along the way  also providing a correction of an algebraic error in \ci{mull:19:1}.

\begin{Lemma}
	\label{lma: mean_uniform}
	Under assumption (A1),
	\begin{equation*}
	\sup_{s \in [0,1]} d(\hat{\mu}(s),\mu(s))=o_P(1).
	\end{equation*}
\end{Lemma}

\begin{proof}
	See Proposition 4 in \ci{mull:19:1} for the proof.
\end{proof}

\begin{Lemma}
	\label{lma:entropy}
	Under assumptions (A2),(A3) and (A4), it holds  for the function class $\mathcal{F}_{\delta}$ as defined in equation \eqref{f_delta}  that  $J_{[]}(1,\mathcal{F}_{\delta},L^{2}(P)) =O(\sqrt{\log{1/\delta}})$ as $\delta \rightarrow 0$.
\end{Lemma}

\begin{proof}
	By definition $\mathcal{F}_\delta$ comprises of functions 
	$f_{\o,s}(x)=d^2(x(s),\o)-d^2(x(s),\mu(s))$. We see that
	\begin{align*}
	& \left| f_{\o,s}(x)-f_{\o^\ast,s^\ast}(x)\right|  \\ & \leq  \left| f_{\o,s}(x)-f_{\o^\ast,s}(x)\right|+ \left| f_{\o^\ast,s}(x)-f_{\o^\ast,s^\ast}(x)\right| \\ & \leq 4M \left( d(\o,\o^*)+  d(x(s),x(s^\ast))+d(\mu(s),\mu(s^\ast))\right).
	\end{align*}
	
	\noindent Since the random functions $X(\cdot)$ have uniformly H\"older continuous trajectories almost surely, Proposition 3 in \cite{mull:19:1} implies that $\mu(\cdot)$ is uniformly continuous. A consequence is that whenever $s$ and $s^\ast$ are close enough, such that $d(\mu(s),\mu(s^\ast))<\rho$, then using the curvature condition in assumption (A2) and the H\"older continuity of $X(\cdot)$ in assumption (A3) one has 
	\begin{align}
	& d^\beta(\mu(s),\mu(s^\ast)) \nonumber  \\ & \leq D^{-1} \lbrace  E(d^2(X(s),\mu(s^\ast)) -  E(d^2(X(s),\mu(s)) + E(d^2(X(s^\ast),\mu(s)) -  E(d^2(X(s^\ast),\mu(s^\ast)) \rbrace  \nonumber \\ & \leq 4 M D^{-1} E(d(X(s),X(s^\ast)) \nonumber \\ & \leq 4MD^{-1} E(G(X)) |s-s^\ast|^\alpha. \nonumber
	\end{align}
	Therefore whenever $s$ and $s^\ast$ are such that $|s-s^\ast|< \left(\frac{\rho^\beta}{4MD^{-1} E(G(X))}\right)^{1/\alpha}$, we have $\alpha/\beta$-H\"older continuity of $\mu(\cdot)$. Therefore, by taking  $K=\{4MD^{-1} E(G(X))\}^{1/\beta}$ one has for all $|s-s^\ast|< \left(\frac{\rho}{K}\right)^{\beta/\alpha}$,
	\begin{equation}
	d(\mu(s),\mu(s^\ast)) \leq K |s-s^\ast|^{\alpha/\beta}.
	\end{equation}
	
	Going back to the bound on 	$\left| f_{\o,s}(x)-f_{\o^\ast,s^\ast}(x)\right|$, by assumption (A3) one has for some constant $U>0$, 
	\begin{align*}
	\left| f_{\o,s}(x)-f_{\o^\ast,s^\ast}(x)\right|   \leq  U \left( d(\o,\o^*)+  G(x) |s-s^\ast|^\alpha +d(\mu(s),\mu(s^\ast))\right)
	\end{align*}
	and therefore for some constant $L > 0$,
	\begin{equation*}
||f_{\o,s}-f_{\o^\ast,s^\ast}||_{L_2} \leq L  \left( d(\o,\o^*)+  |s-s^\ast|^\alpha )+d(\mu(s),\mu(s^\ast)) \right).
	\end{equation*}
	In this proof, since we are looking at the limiting behavior $\delta \rightarrow 0$, one can henceforth let $\delta < \min\left(\rho,1\right)$ without any loss of generality. It follows that for any given $0<u<1$, if we take $s^\ast$ to be such that $|s-s^\ast| <  \left( \frac{ \delta u}{K}\right)^{\frac{\beta}{\alpha}}$, then $d(\mu(s),\mu(s^*))$ can be restricted to be less than or equal to $\delta u$ which means $\mu(s^\ast)$ is contained in $B_\delta (\mu(s))$, and therefore within $B_{2\delta}(\o)$.  Let $s_1,s_2, \hdots, s_K$ be a  $ \left( \frac{ \delta u}{K}\right)^{\frac{\beta}{\alpha}}$-net for $[0,1]$ with the metric $|\cdot|$ and $\o^{s_j}_1,\o^{s_j}_2,\hdots, \o^{s_j}_L$ be a $u$-net for $B_{2\delta}(\mu(s_j))$ with metric $d$. Then for any $s \in [0,1]$ and $\o$ such that $d(\o,\mu(s)) < \delta$, one can find $s_j$ and $\o^{s_j}_k$ such that for some constant $L' > 0$,
	\begin{equation*}
	||f_{\o,s}-f_{\o^{s_j}_k,s_j}||_{L_2} \leq L' u.
	\end{equation*}
	This implies that
	\begin{equation*}
	N(L'u,\mathcal{F}_{\delta},L^2(P)) \leq N\left(\left( \frac{ \delta u}{K}\right)^{\frac{\beta}{\alpha}},[0,1],|\cdot|\right) \sup_{s \in [0,1]} N\left(u,B_{2\delta}(\mu(s)),d \right), 
	\end{equation*}
	where $N(u,\mathcal{F}_{\delta},L^2(P))$ is the covering number, i.e. the minimum number of $L_2$-balls of radius $u$ needed to cover $\mathcal{F}_{\delta}$. One therefore has that
	\begin{align*}
	&\log N(2M\delta\eps,\mathcal{F}_{\delta},L^2(P)) \\  \leq & \log N\left(\left( \frac{ 2M\delta^2\eps}{KL'}\right)^{\frac{\beta}{\alpha}},[0,1],|\cdot|\right) +  \sup_{s \in [0,1]} \log N\left(\frac{2M\delta\eps}{L'},B_{2\delta}(\mu(s)),d \right) \\  \leq  & \log \left(C' (\eps\delta^2)^{-\beta/\alpha}\right)+ \sup_{s \in [0,1]} \log N\left(\frac{2M\delta\eps}{L'},B_{2\delta}(\mu(s)),d \right)
	\end{align*}
	for some constant $C'$.  Finally we can bound the entropy integral as
	\begin{align*}
	& \int_{0}^{1} \sqrt{1+\log N(\eps ||F||_{2},\mathcal{F}_{\delta},L^2(P))} d\eps \\ = & \int_{0}^{1} \sqrt{1+\log N(2M\delta\eps,\mathcal{F}_{\delta},L^2(P))} d\eps \\  \leq & \,\, 1+ \sqrt{\log(C')}+ \int_{0}^{1} \sqrt{-\frac{\beta}{\alpha}\log(\eps\delta^2)}d\eps+\int_{0}^{1} \sup_{s \in [0,1]}  \sqrt{\log N\left(2M\delta\eps,B_{2\delta}(\mu(s)),d \right)} d\eps \\  \leq & \,\, 1+ \sqrt{\log(C')}+ \sqrt{\frac{\beta}{\alpha}} \left \lbrace  \int_{0}^{1} \sqrt{-\log(\eps)} d\eps + \sqrt{-2\log(\delta)} \right \rbrace\\ & + \,\, \int_{0}^{1} \sup_{s \in [0,1]}  \sqrt{\log N\left(2M\delta\eps,B_{2\delta}(\mu(s)),d \right)} d\eps.
	\end{align*}
	Assumption (A4) then implies $ J_{[]}(1,\mathcal{F}_{\delta},L^{2}(P))= O(\sqrt{-\log \delta})$  as $\delta \rightarrow 0$, which completes the proof. 
\end{proof}

\begin{Lemma}
	\label{lma: rate}
	Under assumptions (A1)-(A4),
	\begin{equation*}
	\sup_{s \in [0,1]} d(\hat{\mu}(s),\mu(s))=O_P\left(\left(\frac{\sqrt{\log{n}}}{n}\right)^{\frac{1}{2(\beta-1)}}\right).
	\end{equation*}
\end{Lemma}

\begin{proof}
	For a sequence $\{q_n\}$  define the sets 
	\begin{equation*}
	S_{j,n}(s)= \lbrace \o : 2^{j-1} < q_n d(\o,\mu(s)) \leq 2^j\rbrace .
	\end{equation*}
	Choose $\rho > 0$ to satisfy (A2). For any integer $L$,
	\begin{align}
	& P\left(q_n \sup_{s \in [0,1]} d(\hat{\mu}(s),\mu(s))> 2^L \right) \nonumber \\  \leq &
	P\left(\sup_{s \in [0,1]} d(\hat{\mu}(s),\mu(s)) \geq \rho\right)+ \sum_{j \geq L, 2^j \leq q_n \rho} P\left(\hat{\mu}(s) \in S_{j,n}(s) \ \text{for all} \ s \in [0,1]\right) \nonumber \\  \nonumber  \leq &
	P\left(\sup_{s \in [0,1]} d(\hat{\mu}(s),\mu(s)) \geq \rho\right)\\ &  \quad\quad  + \sum_{j \geq L, 2^j \leq q_n {\rho}} P\left( \sup_{s \in [0,1], \om \o \in S_{j,n}(s)} \left|U_n(\o,s)-U(\o,s)\right| \geq D \frac{2^{\beta(j-1)}}{q^\beta_n}\right), \label{eq: onion}
	\end{align}
	where \eqref{eq: onion} follows from assumption (A2) by observing that for any $s \in [0,1]$, 
	\begin{align*}
	  \sup_{\om \o \in S_{j,n}(s)} \left|U_n(\o,s)-U(\o,s)\right|\\ & \hspace{-3cm}  \geq  \left|\inf_{\om \o \in S_{j,n}(s)} U_n(\o,s)- \inf_{\om \o \in S_{j,n}(s)} U(\o,s)\right| \geq D \frac{2^{\beta(j-1)}}{q^\beta_n}.
	\end{align*}
	The first term in  the last line of  \eqref{eq: onion} goes to zero by the uniform convergence of \F \ mean and for each $j$ in the second term it holds that  $d(\o,\mu(s)) \leq \rho$.  We now use the bound on the entropy integral from Lemma \ref{lma:entropy} which is bounded above by $J\sqrt{\log {1/\delta}}$ for all small enough $\delta > 0$,  where $J > 0$ is a constant. By Theorem 2.14.2 of \cite{well:96},
	\begin{equation}
	\label{eq: tail_bound}
	E\left(\sup_{s \in [0,1], \om d(\o,\mu(s)) < \delta} \left|U_n(\o,s)-U(\o,s)\right|\right) \leq \frac{2M\delta \ J_{[]}(1,\mathcal{F}_{\delta},L_{2}(P))}{\sqrt{n}}.
	\end{equation}

	Using this, the entropy bound and the  Markov inequality,  the second term is upper bounded up to a constant by
	\begin{equation}
	\label{eq: upp}
	\sum_{j \geq L, 2^j \leq q_n \tilde{\alpha}} \frac{2MJ 2^j}{\sqrt{n}q_n} \sqrt{\log \frac{q_n}{2^{j-1}}}\frac{q_n^\beta}{D 2^{\beta(j-1)}}.
	\end{equation}
	Setting $q_n=\left( \frac{n}{\sqrt{\log n}} \right)^{1/2(\beta-1)}$, the series in \eqref{eq: upp} is upper bounded by $O \left( \sum_{j \geq L, 2^j \leq q_n \rho} 2^{(1-\beta)j} \right)$, which can be made arbitrarily small by choosing $L$ and $n$ large. This proves the desired result that $\sup_{s \in [0,1]} d(\hat{\mu}(s),\mu(s))=O_P(q_n^{-1})=O_P\left(\left( \frac{n}{\sqrt{\log n}} \right)^{\frac{-1}{2(\beta-1)}}\right)$.
\end{proof}

\setcounter{figure}{0} \renewcommand{\thefigure}{A.\arabic{figure}} 
\subsection*{A.4 Additional Figures: Chicago Divvy Bike Data} 
\label{supp4}

\begin{figure}[H]
	\centering
	\includegraphics[scale = .5]{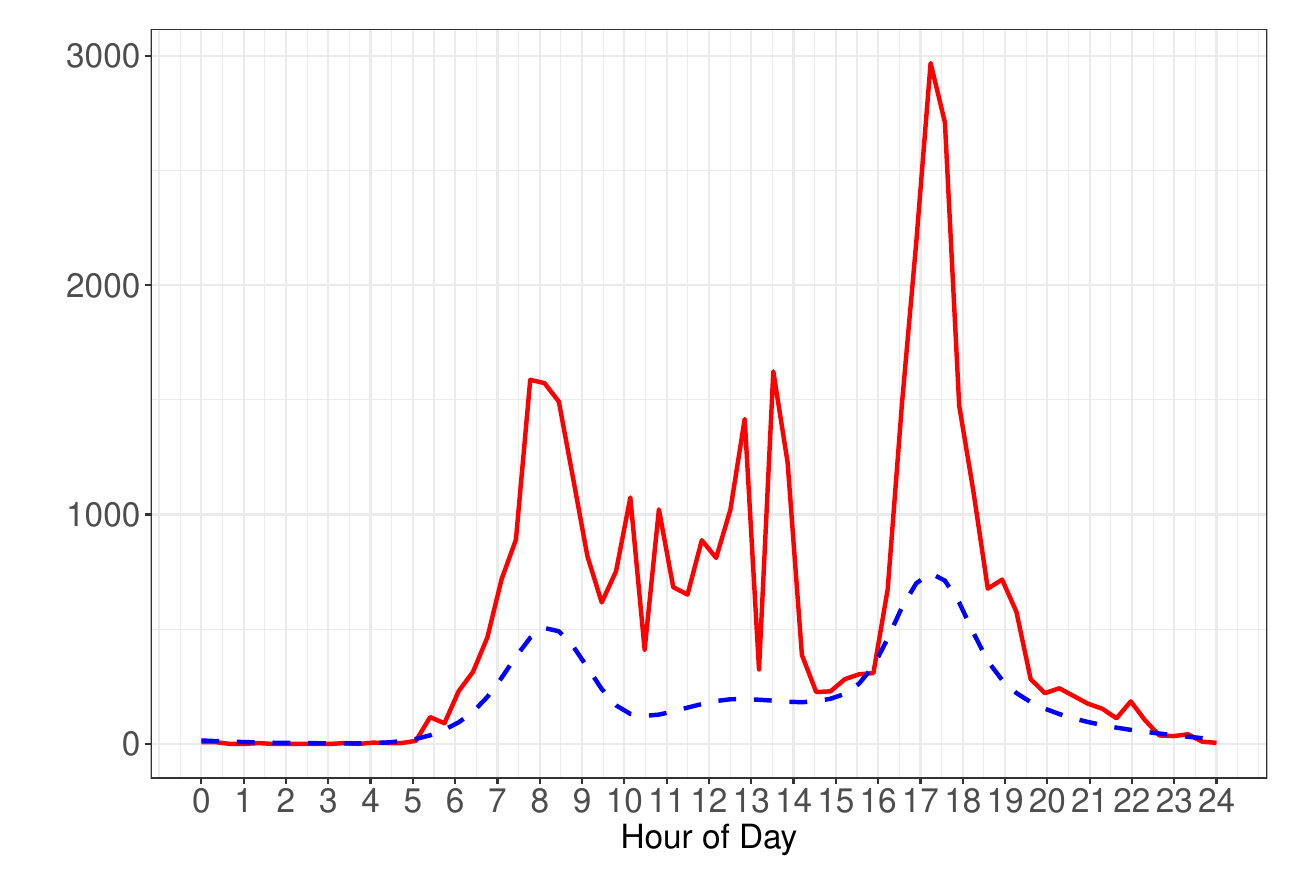}
	\caption{Squared distance trajectory corresponding to 21 August 2017 (solid red) as compared against the overall \F \ variance function (dashed blue).} 
	\label{fig: s1}
\end{figure}

\begin{figure}[H]
	\centering
	\includegraphics[scale = .5]{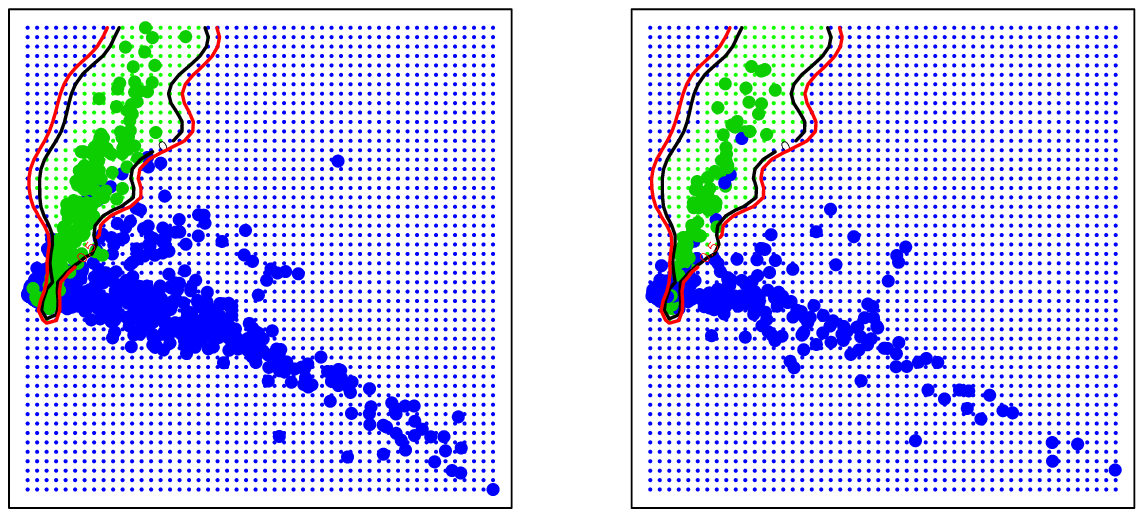}
	\caption{Classification boundary of the  support vector machine (black) and Bayes decision boundary (red),  for the training data (left) and the test data (right).} 
	\label{fig: s2}
\end{figure}

\begin{figure}[H]
	\centering
	\includegraphics[width=\textwidth]{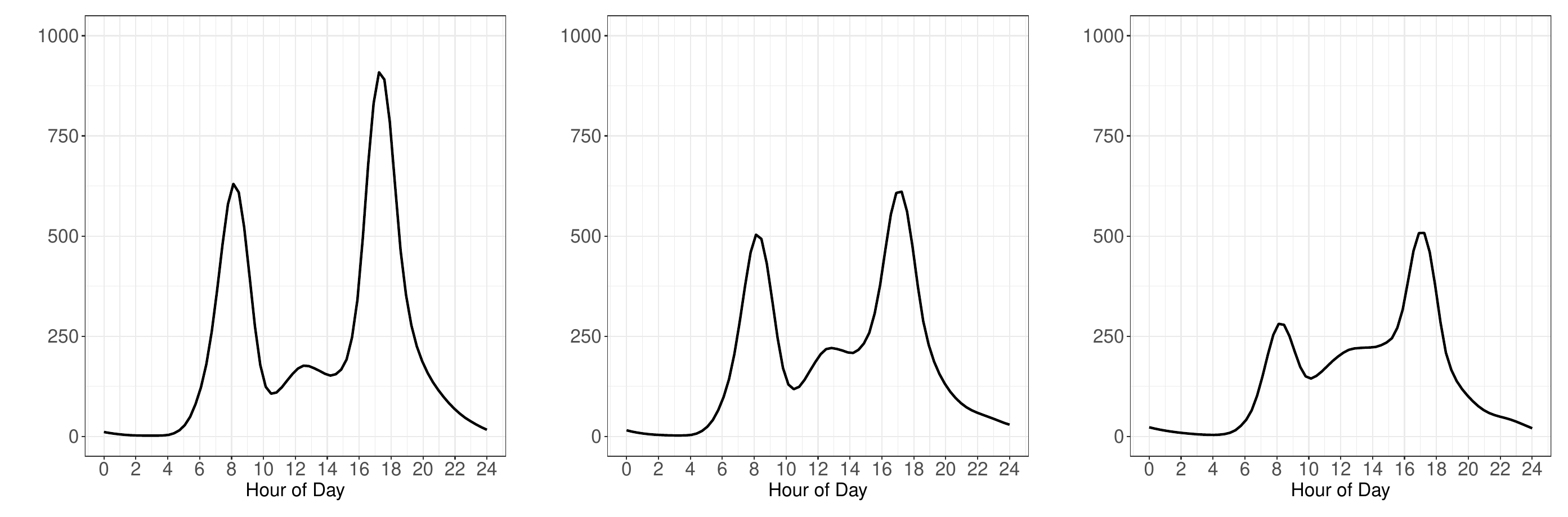}
	\caption{Mean function $\frac{1}{n}\sum_{i=1}^{n} V_i(t)$ of the \F \ variance trajectories of the time courses of daily Divvy bike trip networks in Chicago for weekdays (left), Fridays (middle) and weekends and holidays (right).} 
	\label{fig: s3}
\end{figure}

\begin{figure}[H]
	\centering
	\includegraphics[width=\textwidth]{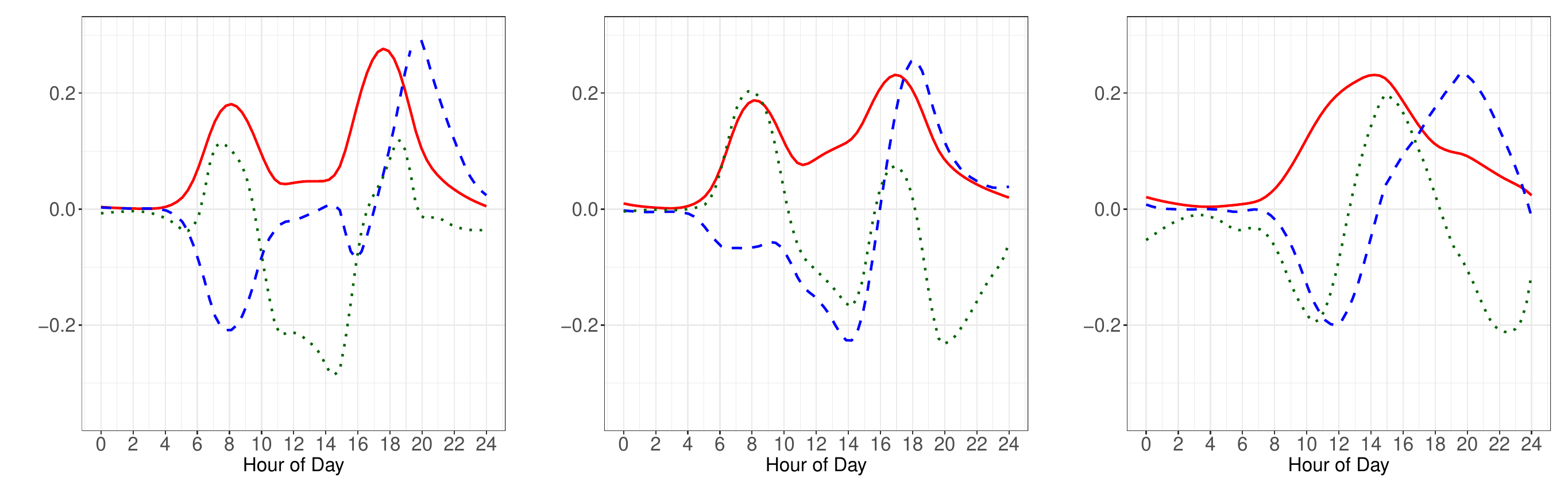}
	\caption{Eigenfunctions corresponding to the FPCA of the daily \F \ variance trajectories for weekdays (left), Fridays (middle) and weekends and holidays (right). The solid red line corresponds to the first, the dashed blue line to the second and the dotted green line to the third eigenfunction.}
	\label{fig: s4}
\end{figure}

\begin{figure}[H]
	\centering
	\includegraphics[width=\textwidth]{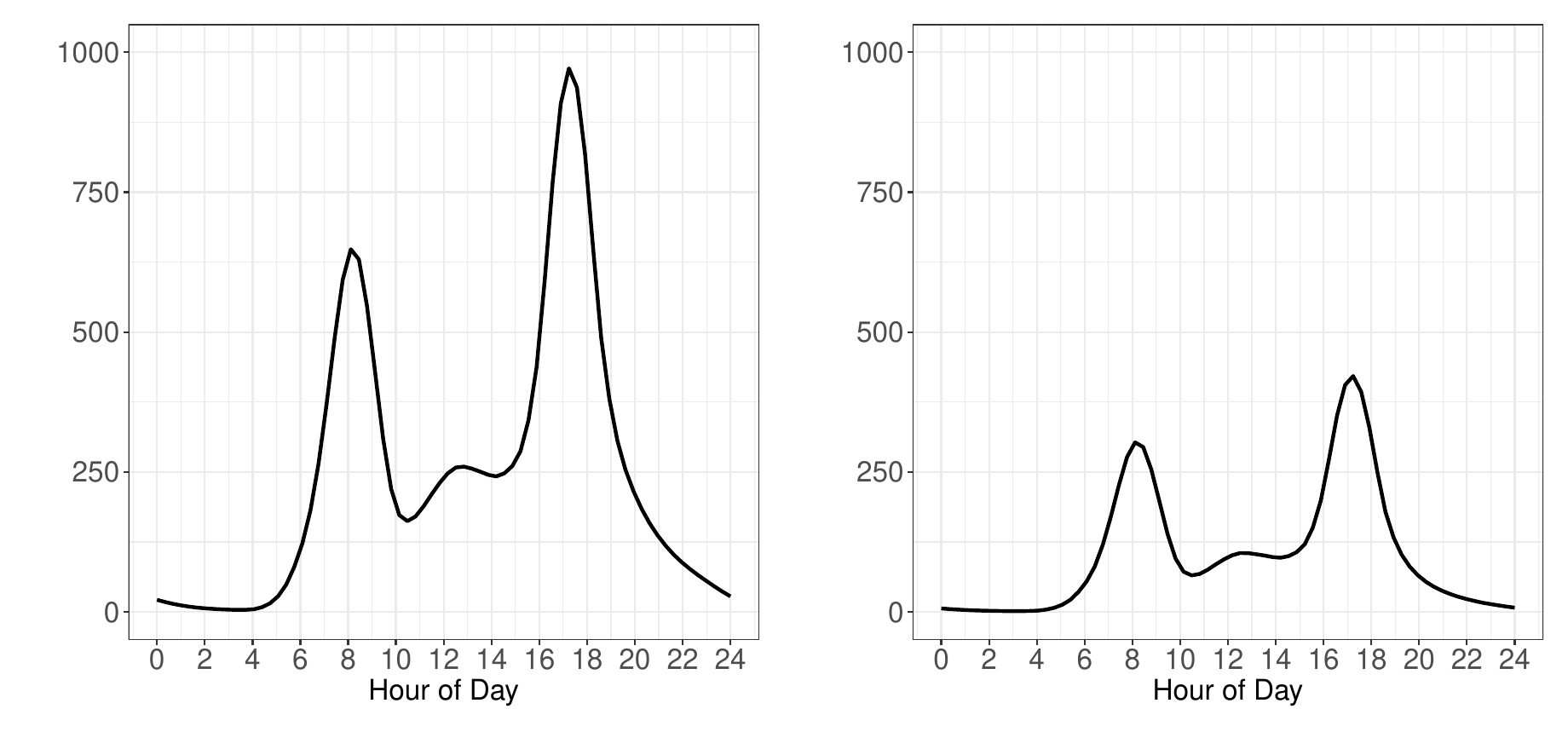}
	\caption{Mean function $\frac{1}{n}\sum_{i=1}^{n} V_i(t)$ of the \F \ variance trajectories of the time courses of daily Divvy bike trip networks in Chicago for Fridays in Spring, Summer and early Fall  (left) and in late Fall and Winter (right).} 
	\label{fig: s5}
\end{figure}

\begin{figure}[H]
	\centering
	\includegraphics[width=\textwidth]{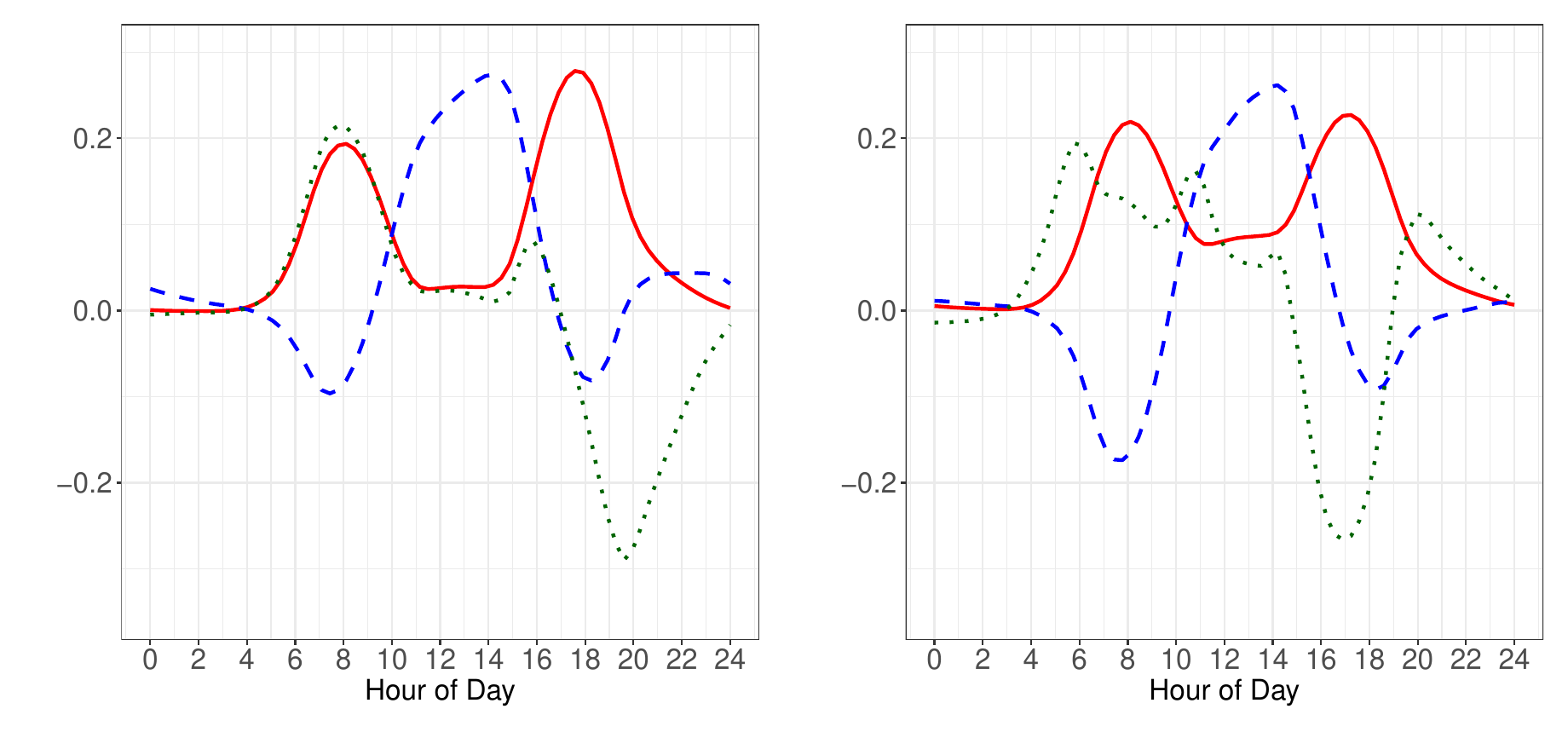}
	\caption{Eigenfunctions corresponding to the FPCA of the daily \F \ variance trajectories for spring, summer and early fall, Fridays (left) and late fall and winter (right). The solid red line corresponds to the first eigenfunction, the dashed blue line to the second and the dotted green line to the third eigenfunction.}
	\label{fig: s6}
\end{figure}

\subsection*{A.5 Discussion on examples of $(\Omega,d)$ satisfying assumptions (A1)-(A5)} 
\label{spaces}

\subsubsection*{Distributions}

Let $\Omega$ be the space of univariate probability distributions represented as quantile functions and let $d$ be the 2-Wasserstein metric. For quantile functions $Q_1(\cdot)$ and $Q_2(\cdot)$, the 2-Wasserstein metric is the $L_2$ metric is 
\begin{equation*}
d^2(Q_1,Q_2) = \int_{0}^1 (Q_1(t)-Q_2(t))^2 dt= ||Q_1-Q_2||^2_{L_2}.
\end{equation*}
Let $Q_{X(s)}$ be the quantile function associated with the distribution of $X(s)$ for a fixed $s$ and $Q_\o$ the quantile function associated with the distribution  $\o \in \O$.  Then \begin{align*}
E\{d^2(X(s),\o)\} = &  E\{||Q_{X(s)}-Q_\o||^2_{L_2}\} \\ = & E\{\langle Q_{X(s)}-E(Q_{X(s)})+E(Q_{X(s)})-Q_\o, Q_{X(s)}-E(Q_{X(s)})+E(Q_{X(s)})-Q_\o \rangle_{L_2}\} \\ = & E\{\langle Q_{X(s)}-E(Q_{X(s)}), Q_{X(s)}-E(Q_{X(s)}) \rangle_{L_2}\}  + \langle E(Q_{X(s)})-Q_\o, E(Q_{X(s)})-Q_\o \rangle_{L_2} \\ = &  E\{||Q_{X(s)}-E(Q_{X(s)})||^2_{L_2}\}+  ||E(Q_{X(s)})-Q_\o||^2_{L_2}.
\end{align*}
Therefore $Q_{\mu(s)} = \operatornamewithlimits{argmin}_{\o \in \O} ||E(Q_{X(s)})-Q_\o||^2_{L_2}=E(Q_{X(s)}) $ by the convexity of the space of quantile functions. In the sample version, let $\bar{Q}(s)(\cdot)= \frac{1}{n} \sum_{i=1}^n Q_{X_i(s)} (\cdot)$. The same  algebra as above implies 
\begin{equation*}
\frac{1}{n} \sum_{i=1}^n d^2(X_i(s),\o)= \frac{1}{n} \sum_{i=1}^n ||Q_{X_i(s)}-\bar{Q}(s)||^2_{L_2}+ ||\bar{Q}(s)-Q_\o||^2_{L_2},
\end{equation*} 
whence  $Q_{\hat{\mu}(s)}=\argmin_{\o \in \O} ||\bar{Q}(s)-Q_\o||^2_{L_2}= \bar{Q}(s)$ by the convexity of the space of quantile functions. Hence both $\mu(s)$ and $\hat{\mu}(s)$ exist and are unique as their quantile representations exist and are unique. Moreover,  the above calculations also imply that 
\begin{equation*}
E\{d^2(X(s),\o)\} - E\{d^2(X(s),\mu(s))\} = ||Q_\o-Q_{\mu(s)}||^2_{L_2}
\end{equation*}
and 
\begin{equation*}
\frac{1}{n} \sum_{i=1}^n d^2(X_i(s),\o)- \frac{1}{n} \sum_{i=1}^n d^2(X_i(s),\hat{\mu}(s)) = ||Q_\o-Q_{\hat{\mu}(s)}||^2_{L_2},
\end{equation*}
whence 
\begin{equation*}
\inf_{s \in [0,1]} \inf_{\om d(\o,\mu(s)) > \epsilon} E(d^2(X(s),\o))-E(d^2(X(s),\mu(s))) > 0.
\end{equation*}
Therefore (A1) holds  with $\tau(\epsilon)=\epsilon^2$ and (A2) holds with $D=1$ and $\beta=2$, while  (A4) is satisfied using the arguments in the proof of Proposition 1 of \cite{mull:19:3}. 

Assumption (A3) is a common requirement for classical real-valued functional data, where trajectories are sometimes  assumed to be twice differentiable with bounded second derivative. An example of time varying distributions is a sample of random density valued trajectories, where  $X(t)$ is a normal distribution for all $t$. For an explicit sample construction, assume that  the  means of these normal distributions are  $\mu(t)$ and the variances are  $\sigma^2(t)$,  both taken as random functions 
\begin{equation*}
\mu(t)= X_1+ X_2 \sin(2 \pi t) \quad \text{and} \quad \sigma^2(t)= X_3 e^{0.25 t},
\end{equation*}
where $X_1$ and $X_2$ are independent uniform r.v.s  in $[-1,1]$ and $X_3$ is an independent exponential r.v., with mean one, truncated to lie in $[0.25,3]$. Each $\o(\cdot)$ in the sample space of trajectories is therefore a density valued process, where $\o(s)$ is a normal distribution with mean of the form $a+b \sin{2\pi s}$, where $a,b \in [-1,1]$,  and variance of the form $c e^{0.25s}$,  where $c \in [0.25,3]$. For two univariate normal distributions, $f_1 \sim N(\mu_1,\sigma^2_1)$ and $f_2 \sim N(\mu_2,\sigma^2_2)$, the Wasserstein distance between $f_1$ and  $f_2$ can be explicitly calculated as 
\begin{equation*}
d(f_1,f_2)= \sqrt{(\mu_1-\mu_2)^2+(\sigma_1-\sigma_2)^2},
\end{equation*}
which for the example above yields 
\begin{equation}
d(X(s),X(t))= \sqrt{X_2^2 (\sin(2 \pi s)-\sin(2 \pi t))^2+ X_3^2 (e^{0.25s}-e^{0.25t})^2}.
\end{equation}
Condition (A3) is satisfied as both $\sin 2 \pi t$ and $e^{0.25t}$ have bounded derivatives in $[0,1]$ and the second moments of $X_1$, $X_2$ and $X_3$ are finite.

\subsubsection*{Networks}

Let $\Omega$ be the space of graph adjacency matrices or graph Laplacians of weighted networks, with uniformly bounded weights and let $d$ be the Frobenius metric. For matrices $U$ and $V$ the Frobenius metric is given by
\begin{equation*}
d^2(U,V) = \sum_{i,j} (u_{ij}-v_{ij})^2=(U-V)^T(U-V).
\end{equation*}
For any $\o \in \O$, by the properties of the Frobenius metric,
\begin{align*}
E\{d^2(X(s),\o)\} = &  E\{(X(s)-\o)^T(X(s)-\o)\} \\ = & E\{(X(s)-E(X(s))+E(X(s))-\o)^T(X(s)-E(X(s))+E(X(s))-\o)\} \\ = & E\{(X(s)-E(X(s)))^T(X(s)-E(X(s)))\} + (E(X(s))-\o)^T(E(X(s))-\o) \\ = & E\{d^2(X(s),E(X(s)))\} + d^2(E(X(s)),\o).
\end{align*}
Therefore $\mu(s)=\operatornamewithlimits{argmin}_{\o \in \O} E\{d^2(X(s),\o)\} $ = $\operatornamewithlimits{argmin}_{\o \in \O} d^2(E(X(s)),\o)=E(X(s)) $ by the convexity of the space of graph Laplacians and adjacency matrices. In the sample version, let $\bar{X}(s)= \frac{1}{n} \sum_{i=1}^n X_i(s)$. By following the above algebra, one can show that
\begin{equation*}
\frac{1}{n} \sum_{i=1}^n d^2(X_i(s),\o)= \frac{1}{n} \sum_{i=1}^n d^2(X_i(s),\bar{X}(s))+ d^2(\bar{X}(s),\o),
\end{equation*} 
which implies that $\hat{\mu}(s)= \argmin_{\o \in \O} \frac{1}{n} \sum_{i=1}^n d^2(X_i(s),\o)= \argmin_{\o \in \O} d^2(\bar{X}(s),\o)= \bar{X}(s)$, again by the convexity of the space of graph Laplacians and adjacency matrices. Hence both $\mu(s)$ and $\hat{\mu}(s)$ exist and are unique. Moreover,  the above calculations imply that
\begin{equation*}
E\{d^2(X(s),\o)\} - E\{d^2(X(s),\mu(s))\} = d^2(\o,\mu(s))
\end{equation*}
and 
\begin{equation*}
\frac{1}{n} \sum_{i=1}^n d^2(X_i(s),\o)- \frac{1}{n} \sum_{i=1}^n d^2(X_i(s),\hat{\mu}(s)) = d^2(\o,\hat{\mu}(s)),
\end{equation*}
whence
	\begin{equation*}
\inf_{s \in [0,1]} \inf_{\om d(\o,\mu(s)) > \epsilon} E(d^2(X(s),\o))-E(d^2(X(s),\mu(s))) > 0
\end{equation*}
and therefore (A1) holds with $\tau(\epsilon)=\epsilon^2$ and (A2) holds with $D=1$ and $\beta=2$. As  the graph adjacency matrices and graph Laplacians of networks with bounded edge connectivities form a subset of a larger finite-dimensional bounded Euclidean space, (A4) is also satisfied  using the arguments in the proof of Proposition 2 of \cite{mull:19:3}. 

In our example of graph Laplacians or graph adjacencies of weighted networks with a fixed number of nodes, assumption (A3) on H\"older continuity translates into  H\"older equicontinuity of the individual edge connectivities. For an example of a practical data generating model that satisfies this criterion, assume that  the network adjacency matrix $X(t)$ is generated as 
\begin{equation*}
\{X(t)\}_{k,l} = W_{kl}(t)  \left \lbrace \frac{1+\sin(V_{kl}\pi(t+U_{kl}))}{2} \right \rbrace,
\end{equation*}
where $W_{kl}(t)$ is a smooth weight function with bounded derivatives, $U_{kl}$ is generated independently from $U(0,1)$ and $V_{kl}$ is generated uniformly from $\{1,2,\dots,10\}$. Here $U_{kl}$ and $V_{kl}$ are the random phase and frequency shift of the sine function which controls where in the time domain and how many times the edge connectivity $\{X(t)\}_{k,l}$ hits zero, in which case there is no  link between nodes $k$ and $l$. Each network in the sample space, with adjacency matrix  $\o(t)$, is of the form 
\begin{equation*}
\{\o(t)\}_{k,l} = W_{kl}(t)  \left \lbrace \frac{1+\sin(v\pi(t+u))}{2} \right \rbrace.
\end{equation*}
Therefore all the networks in the sample space satisfy assumption (A3). We also adopt this construction  as a data generating model in our simulations. 
\ed